\documentclass{article}
\title{Between hard and soft thresholding: optimal iterative thresholding algorithms}
\author{Haoyang Liu and Rina Foygel Barber}
\usepackage{amsmath,amsthm,amsfonts,amssymb,natbib,graphicx}
\usepackage[margin=1.3in]{geometry}
\newtheorem*{theorem*}{Theorem}
\newtheorem{theorem}{Theorem}
\newtheorem{lemma}{Lemma}

\usepackage{hyperref}

\newcommand{\reg}{\mathsf{R}}
\newcommand{\loss}{\mathsf{f}}
\newcommand{\thr}{\Psi_s}
\newcommand{\mthr}{\widetilde\Psi_s}
\newcommand{\thrsig}{\Psi_{s;\sigma}}
\newcommand{\mthrsig}{\widetilde\Psi_{s;\sigma}}
\newcommand{\RC}{\gamma}
\newcommand{\mRC}{\widetilde\gamma}
\newcommand{\RTc}{\thr^{\textnormal{RT},c}}

\newcommand{\RTone}{\thr^{\textnormal{RT},1}}
\newcommand{\RTzero}{\thr^{\textnormal{RT},0}}
\newcommand{\RT}{\thr^{\textnormal{RT}}}
\newcommand{\mRT}{\mthr^{\textnormal{RT}}}
\newcommand{\LQ}{\thr^{\ell_q}}

\newcommand{\LQuniv}{\thr^{\ell_{2/3}}}
\newcommand{\HT}{\thr^{\textnormal{HT}}}
\newcommand{\mHT}{\mthr^{\textnormal{HT}}}
\newcommand{\ST}{\thr^{\textnormal{ST}}}
\newcommand{\R}{\mathbb{R}}
\newcommand{\inner}[2]{\langle #1 , #2 \rangle}
\newcommand{\norm}[1]{\| #1 \|}
\newcommand{\fronorm}[1]{\| #1 \|_{\textnormal{F}}}
\newcommand{\Fronorm}[1]{\left\| #1 \right\|_{\textnormal{F}}}

\newcommand{\rank}{\textnormal{rank}}
\newcommand{\Xcal}{\mathcal{X}}
\newcommand{\Ccal}{\mathcal{C}}

\begin{document}

\maketitle

\begin{abstract}
Iterative thresholding algorithms seek to optimize a differentiable objective function over a sparsity or rank constraint by alternating between gradient steps that reduce the objective, and thresholding steps that enforce the constraint. This work examines the choice of the thresholding operator, and asks whether it is possible to achieve stronger guarantees than what is possible with hard thresholding. We develop the notion of relative concavity of a thresholding operator, a quantity that characterizes the worst-case convergence performance of any thresholding operator on the target optimization problem. Surprisingly, we find that commonly used thresholding operators, such as hard thresholding and soft thresholding, are suboptimal in terms of worst-case convergence guarantees. Instead, a general class of thresholding operators, lying between hard thresholding and soft thresholding, is shown to be optimal with the strongest possible convergence guarantee among all thresholding operators. Examples of this general class includes $\ell_q$ thresholding with appropriate choices of $q$, and a newly defined {\em reciprocal thresholding} operator. We also investigate the implications of the improved optimization guarantee in the statistical setting of sparse linear regression, and show that this new class of thresholding operators attain the optimal rate for computationally efficient estimators, matching the Lasso.
\end{abstract}

\section{Introduction}\label{sec:intro}

We consider the general problem of sparse optimization, where we seek to optimize a likelihood function or loss function subject to a sparsity constraint,
\[\min_{x\in\R^d, \norm{x}_0\leq s}\loss(x).\]
Here $\loss:\R^d\rightarrow \R$ is the target function that we would like to minimize, while the constraint $\norm{x}_0\leq s$ requires that the solution vector $x$ has at most $s$ many nonzero entries. Similarly, we may work with a matrix parameter $X\in\R^{n\times m}$ and search for a low-rank solution,
\[\min_{X\in\R^{n\times m},\rank(X)\leq s}\loss(X).\]
Optimization problems over a sparsity constraint or a rank constraint are ubiquitous in high-dimensional statistics and machine learning. Sparsity of a vector parameter $x$ represents the idea that we can model the data using a small fraction of the available features, which, for instance, may correspond to covariates in a regression model or to basis expansion terms in a nonparametric function estimation problem. Similarly, a rank constraint on a matrix parameter $X$ might correspond to an underlying factor model with a small number of factors. We will focus on problems where $\loss$ is a differentiable function, as is often the case for many likelihood models and other loss functions.

In this work, we will study the iterative thresholding approach, where gradient steps that lower the value of the target function $\loss$ are alternated with thresholding steps to enforce the sparsity constraint---for instance, {\em hard thresholding} sets all but the largest $s$ entries to zero, while {\em soft thresholding} shrinks all values towards zero equally until the sparsity constraint is satisfied. (The same ideas apply to a rank constraint, by thresholding or shrinking singular values instead of vector entries. For simplicity, we will primarily discuss the sparse minimization problem, and will return to the low-rank problem later on.)

For sparse minimization of a differentiable target function $\loss(x)$, many existing algorithms can be broadly described as iterating steps of the following form:
\begin{equation}\label{eqn:iter_intro}\begin{cases}
\text{Gradient step: $x'_t = x_{t-1} - \eta_t \cdot \nabla \loss(x_{t-1})$ for some step size $\eta_t$},\\
\text{Sparsity step: $x_t$ = some sparse (or nearly sparse) approximation to $x'_t$}.
\end{cases}\end{equation}
Our aim in this work is to characterize the type of thresholding operators that are likely to be most successful at converging to a good solution, i.e.~to a value of $\loss(x)$ that is as low as possible. Is an iterative thresholding algorithm most likely to succeed if we use hard thresholding, soft thresholding, or yet another form of thresholding to enforce the sparsity constraint?

In this work, we find that the worst-case performance of a thresholding operator, relative to a broad class of target functions $\loss$ that we may want to minimize, is fully characterized by a simple measure that we call the {\em relative concavity}. The relative concavity studies the behavior of the sparse thresholding map $x'_t\mapsto x_t$ in the iterative algorithm~\eqref{eqn:iter_intro}, viewed as an approximate projection onto the space of $s$-sparse vectors. Using relative concavity as a tool to evaluate and compare different thresholding operators, we find that commonly used thresholding operators, for example hard thresholding and soft thresholding, are indeed suboptimal. Instead, we characterize a general class of thresholding operators, lying between hard thresholding and soft thresholding, that we show to be optimal. This class includes $\ell_q$ norm thresholding, where $q\in(0,1)$ is chosen adaptively relative to the particular problem; furthermore, choosing $q=2/3$ is ``universal'' in the sense that it is nearly optimal across all sparse thresholding problems. We also develop the {\em reciprocal thresholding} operator, which enjoys the same optimality guarantees as $\ell_q$ thresholding, but with a closed-form equation for the iterative thresholding step. These simple and efficient iterative thresholding methods are then applied to the statistical setting of sparse linear regression problem:
\begin{equation}
y=X\theta_0+z,
\end{equation}
and are shown to match the Lasso in terms of the resulting guarantee on estimating the true mean vector $X\theta_0$.

\section{Background: sparse minimization}

Before defining relative concavity and the reciprocal thresholding operator, we first review some of the recent literature on hard thresholding and related methods, and define the convexity and smoothness properties of the objective function $\loss$ that we will assume throughout this work.

\subsection{Restricted strong convexity and restricted smoothness}
In many problems in high-dimensional statistics, we aim to optimize loss functions that may be very poorly conditioned in general, but nonetheless exhibit convergence properties of a well-conditioned function when working only with sparse or approximately sparse vectors. This behavior is captured in the notions of restricted strong convexity and restricted smoothness (see e.g.~\citet{negahban2009unified,loh2013regularized} for background).

A differentiable function $\loss:\R^d\rightarrow\R$ satisfies {\em restricted strong convexity}
with parameter $\alpha$ at sparsity level $s$, abbreviated as $(\alpha,s)$-RSC, if
\[\loss(y)\geq \loss(x) + \inner{\nabla \loss(x)}{y-x} + \frac{\alpha}{2}\norm{x-y}^2_2\text{ for all $x,y\in\R^d$ with $\norm{x}_0\leq s,\norm{y}_0\leq s$.}\]
Similarly, $\loss$ satisfies {\em restricted smoothness}
with parameter $\beta$ at sparsity level $s$, abbreviated as $(\beta,s)$-RSM, if
\[\loss(y)\leq \loss(x) + \inner{\nabla \loss(x)}{y-x} + \frac{\beta}{2}\norm{x-y}^2_2\text{ for all $x,y\in\R^d$ with $\norm{x}_0\leq s,\norm{y}_0\leq s$.}\]
Our results will focus on $\kappa = \beta/\alpha$, the {\em condition number} of the function $\loss$ (at the given sparsity level $s$).

\subsection{Iterative hard thresholding}\label{sec:IHT}
Recent work by \citet{jain2014iterative} studies the {\em iterative hard thresholding algorithm},
which alternates between taking a gradient step, $x - \eta\nabla\loss(x)$,
and projecting onto the sparsity constraint. Specifically,
given a target sparsity level $s$ and an initial point $x_0\in\R^d$,
the iterative
step of the algorithm is defined by
\begin{equation}\label{eqn:IHT}
x_t = \HT\big(x_{t-1} - \eta \nabla \loss(x_{t-1})\big),
\end{equation}
where $\HT$ is the ``hard thresholding'' operator, which truncates any vector $z\in\R^d$ to its $s$ largest entries,
\[\big(\HT(z)\big)_i = \begin{cases}
z_i, & i\in S, \\ 0, & i\not\in S,
\end{cases}\]
where $S\subset\{1,\dots,d\}$ indexes the $s$ largest-magnitude entries of $z$.\footnote{To be fully precise, in the case of a tie between different entries of $z$, we may need to choose which entries to keep and which to set to zero. This choice will not matter from the point of view of our theoretical analysis, and from this point on, we will assume that we have fixed some map $z\mapsto S$, mapping each vector $z\in\R^d$ to a set $S\subset\{1,\dots,d\}$ corresponding to the indices of the $s$ largest entries, so that $|S|=s$ and $\min_{i\in S}|z_i| \geq \max_{j\not\in S}|z_j|$, for every $z$. For instance, in the case of a tie between $z_i$ and $z_j$ for the position of the $s$th largest-magnitude entry, we might follow the rule that we choose to keep entry $i$ if $i<j$ and to keep entry $j$ otherwise. Since the exact choice of the rule for breaking ties is not relevant for our results here, we will implicitly assume it to be fixed for the remainder of this paper.}

\paragraph{Restricted optimality for iterative hard thresholding} It is well known that, due to the nonconvexity of the sparsity constraint $\norm{x}_0\leq s$, the iterative hard thresholding algorithm cannot be guaranteed to find the global minimum, $\min_{\norm{x}_0\leq s}\loss(x)$---at least, not without strong assumptions. In other words, it may be the case that $\lim_{t\rightarrow \infty}\loss(x_t)$ is strictly larger than $\min_{\norm{x}_0\leq s}\loss(x)$.
However, \citet{jain2014iterative}'s analysis of the iterative hard thresholding algorithm~\eqref{eqn:IHT} proves that IHT achieves a weaker optimization guarantee, converging to a loss value that is at least as small as the best value attained under a more restricted constraint $\norm{x}_0\leq s'$ where $s'<s$. More precisely, \citet[Theorem 1]{jain2014iterative} prove that, for an objective function $\loss$ satisfying $(\alpha,s)$-RSC and $(\beta,s)$-RSM,
\begin{equation}\label{eqn:jain_IHT}
\loss(x_t) \leq \min_{\norm{y}_0\leq s/(32\kappa^2)}\left\{ \loss(y) + \left(1 - \frac{1}{12\kappa}\right)^t\cdot (\loss(x_0)-\loss(y))\right\},
\end{equation}
where $\kappa=\beta/\alpha$, and where the step size is taken to be $\eta \propto 1/\beta$. In other words, their result proves linear convergence to the bound
\[\lim_{t\rightarrow \infty}\loss(x_t) \leq \min_{\norm{y}_0\leq s/(32\kappa^2)}\loss(y),\]
meaning that while IHT may not find the global minimum of $\loss(x)$ relative to the $s$-sparsity constraint, it is nonetheless guaranteed to perform at least as well as the best $s/(32\kappa^2)$-sparse solution. An analogous result is proved for the low-rank setting, thresholding singular values instead of vector entries.

In this work, we will refer to this type of result as a {\em restricted optimality} guarantee, where the output of an $s$-sparse optimization algorithm is guaranteed to perform well relative to a more restrictive $s'$-sparsity constraint, for some $s'<s$. In particular, we will be interested in the sparsity ratio $s'/s$---the ratio between the sparsity level $s$ used in the algorithm, versus the level $s'$ appearing in the guarantee. Ideally, we would like this ratio to be as close to $1$ as possible, for the strongest possible guarantee.

\subsection{Related literature}

\paragraph{Iterative thresholding}
There exists a vast literature on the properties of iterative thresholding algorithms, especially iterative hard thresholding, regarding the optimization properties and statistical guarantees of these algorithms. Recent results in this area include the work of \citet{blumensath2009iterative,jain2014iterative,chen2015fast,bhatia2015robust,jain2016structured,cai2016optimal,kyrillidis2014matrix}.

Accelerated forms of the iterative hard thresholding algorithm are studied in \citet{kyrillidis2011recipes,blumensath2012accelerated,khanna2017iht}. In particular, \citet{khanna2017iht} finds substantial theoretical and empirical improvement
over the original non-accelerated version of the algorithm. \citet{nguyen2017linear} studies iterative hard thresholding in the context of stochastic gradient descent, where at each step $t$ we only have access to a noisy vector that approximates the true current gradient, $\nabla\loss(x_t)$. The works mentioned here also consider thresholding algorithms for the low-rank setting, truncating singular values instead of vector entries. More broadly, \citet{nguyen2017linear}'s work considers approximate thresholding procedures and more general definitions of sparsity.

To the best of our knowledge, the question of optimality among thresholding operators has not been addressed before, and it is the goal of this work to provide a framework to identify the worst-case convergence behavior of {\em all} thresholding operators and to find the ones that enjoy the optimal restricted optimality guarantee.

\paragraph{Penalized and constrained optimization methods}
The sparse optimization problem can alternately be approximated by a penalized minimization problem,
\[\min_{x\in\R^d}\left\{\loss(x) + \lambda\reg(x)\right\},\]
or a constrained optimization problem,
\[\min_{x\in\R^d}\left\{\loss(x) \ : \ \reg(x) \leq c\right\},\]
where $\reg(x)$ is a sparsity-promoting regularizer, and $\lambda$ and $c$ are tuning parameters controlling the penalization or constraint. Of course, choosing $\reg(x) = \norm{x}_0$ would reduce to the original target optimization problem, but these minimizations are generally only feasible to solve if $\reg(x)$ is some relaxation of the sparsity constraint/penalty. For example, the Lasso~\citep{tibshirani1996regression} uses a convex regularizer, $\reg(x)=\norm{x}_1$, which enjoys many strong guarantees of accurate estimation of the true sparse signal $x$ and of its support. More recently, many nonconvex penalties have been proposed that reduce the shrinkage bias of the Lasso, at the cost of a more challenging optimization problem, such as the SCAD~\citep{fan2001variable} and MCP~\citep{zhang2010nearly} penalties. The $\ell_q$ norm, for $q\in(0,1)$, has also been extensively studied as a compromise between the convex but biased $\ell_1$ norm (as in the Lasso), and the theoretically optimal but computationally infeasible $\ell_0$ norm (i.e.~the sparsity constraint, $\norm{x}_0\leq s$). Results for the $\ell_q$ norm include work by \citet{chartrand2007exact,foucart2009sparsest,kabashima2009typical,lai2011unconstrained}. \citet{zheng2015does}'s recent work studies the $\ell_q$ norm using the framework of approximate message passing to characterize its superior performance relative to the convex $\ell_1$ norm. While the resulting optimization problem is nonconvex for these alternatives to the $\ell_1$ norm, \citet{loh2013regularized} show that restricted strong convexity in the objective function $\loss$ is sufficient to outweigh bounded concavity in the penalty, to ensure successful optimization within a small error tolerance.

The penalized or constrained formulations of the sparse minimization problem may initially appear to be fundamentally different from the iterative thresholding approach. However, these penalized or constrained problems are often optimized with proximal gradient descent or projected gradient descent algorithms---specifically, for a penalty, the proximal gradient descent algorithm iterates the steps
\[\begin{cases}
\textnormal{Gradient step: }x'_t = x_{t-1} - \eta_t\cdot\nabla\loss(x_{t-1})\text{ for some step size $\eta_t$},\\
\textnormal{Proximal step: }x_t =\arg\min_{x\in\R^d} \left\{\frac{1}{2}\norm{x - x'_t}^2_2 + \eta_t \lambda \reg(x)\right\},\end{cases}\]
while for a constraint, projected gradient descent iterates the steps
\[\begin{cases}
\textnormal{Gradient step: }x'_t = x_{t-1} - \eta_t\cdot\nabla\loss(x_{t-1})\text{ for some step size $\eta_t$},\\
\textnormal{Projection step: }x_t =\arg\min_{x\in\R^d} \left\{\frac{1}{2}\norm{x - x'_t}^2_2 \ : \ \reg(x)\leq c\right\}.\end{cases}\]
Since $\reg(x)$ is a sparsity-promoting regularizer, each iteration $x_t$ will therefore be sparse or approximately sparse. In this way, the penalized loss or constrained loss formulations of the sparse minimization problem can be viewed as analogous to the family of iterated thresholding algorithms, where the thresholding step is replaced by penalizing or constraining a regularizer $\reg(x)$ that is a relaxation of the sparsity constraint. (We will discuss the regularized problem more in Section~\ref{sec:softthresh}.)

\section{Convergence of iterative thresholding}\label{sec:iterative_thresholding}

In this section, we examine the performance of gradient
descent with iterative thresholding, for various choices of the thresholding
operator $\thr$. Specifically, after initializing at any point $x_0\in\R^d$,
the algorithm proceeds by alternating between taking a gradient descent step,
and applying a thresholding operator:
\begin{equation}\label{eqn:iterative_thresh}
x_t = \thr\big(x_{t-1} - \eta_t \nabla \loss(x_{t-1})\big),
\end{equation}
where $\thr: \R^d \rightarrow \{x\in\R^d:\norm{x}_0\leq s\}$ is some thresholding operator that enforces $s$-sparsity at each step.

\paragraph{Step size choice}
Throughout the paper, we will primarily study this generalized iterative thresholding algorithm
under the choice of a universal fixed step size $\eta = 1/\beta$, where
$\beta$ is the restricted smoothness parameter for the function $\loss$. When $\beta$ is unknown, we will also consider the following adaptive choice of step size based on exact line search:
\begin{equation}\label{eqn:iterative_thresh_eta_adaptive}
\begin{cases}
\text{Define }\widetilde{x}_t(\eta) = \thr\big(x_{t-1} - \eta \nabla \loss(x_{t-1})\big),\\
\text{Choose }\eta_t = \max\left\{\eta\geq 0 : \loss(\widetilde{x}_t(\eta))\leq \loss(x_{t-1}) + \inner{\widetilde{x}_t(\eta) - x_{t-1}}{\nabla\loss(x_{t-1})} + \frac{1}{2\eta}\norm{\widetilde{x}_t(\eta) - x_{t-1}}^2_2\right\},\\
\text{Set }x_t = \widetilde{x}_t(\eta_t).
\end{cases}
\end{equation}
Note that, since $x_{t-1}$ and $\widetilde{x}_t(\eta)$ are both $s$-sparse, the curvature condition
\begin{equation}\label{eqn:stepsize_check}
\loss(\widetilde{x}_t(\eta))\leq \loss(x_{t-1}) + \inner{\widetilde{x}_t(\eta) - x_{t-1}}{\nabla\loss(x_{t-1})} + \frac{1}{2\eta}\norm{\widetilde{x}_t(\eta) - x_{t-1}}^2_2\
\end{equation}
is necessarily satisfied for any $\eta\leq \frac{1}{\beta}$ due to the restricted smoothness property. Therefore we will always have $\eta_t \geq \frac{1}{\beta}$. Intuitively, the rule not only helps us get rid of the need to know $\beta$, but also allows the algorithm to take larger step size for more progress when possible. In practice, we would consider using a backtracking line search, that is, starting from a large step size and iteratively shrinking it until condition~\eqref{eqn:stepsize_check} is satisfied. In this way, condition~\eqref{eqn:stepsize_check} is similar to the classical Armijo rule for backtracking line search. For simplicity of our theoretical result we do not treat inexact linesearch in the following.

\paragraph{Restricted optimality}
Given an iterative algorithm that keeps the sparsity of the iterations at $s$, as discussed in Section~\ref{sec:IHT}, we cannot hope to achieve {\em global optimality} (i.e.~a guarantee that $\loss(x_t)$ is nearly as good as the best $s$-sparse solution, $\min_{\norm{x}_0\leq s}\loss(x)$), but we can instead prove guarantees of {\em restricted optimality}, that is $\lim_{t\rightarrow\infty}\loss(x_t)\leq \min_{\norm{x}_0\leq s'}\loss(x)$, for some tighter sparsity constraint $s'\leq s$. We will assess a thresholding operator $\thr$ based on its ability to guarantee restricted optimality relative to a sparsity level $s'$ that is as close to $s$ as possible, i.e.~a sparsity ratio $\rho=s'/s$ that is as close to $1$ as possible.

\subsection{Relative concavity of a thresholding operator}
Let $s\in\{1,\dots,d\}$ be any fixed sparsity level and let $\rho\in[0,1]$.
We define the {\em relative concavity} of an $s$-sparse thresholding
operator $\thr$ relative to sparsity proportion $\rho$ as
\[\RC_{s,\rho}(\thr) = \sup\left\{\frac{\inner{y-\thr(z)}{z-\thr(z)}}{\norm{y-\thr(z)}^2_2} \ : \ y,z\in\R^d, \  \norm{y}_0\leq \rho s, \ y\neq \thr(z)\right\}.\]
Note that $\frac{\inner{y-\thr(z)}{z-\thr(z)}}{\norm{y-\thr(z)}^2_2}$ is the coefficient of projection when projecting $z-\thr(z)$ onto $y-\thr(z)$, and measures how much these two vectors align. To understand the term ``relative concavity'' in the name,
we note that if $\thr$ were a projection operator to some convex constraint set $\mathcal{C}$, then we would have $\inner{y-\thr(z)}{z-\thr(z)}\leq 0$
for any $y\in\mathcal{C}$, by the properties of convex projections.
For sparse estimation, the constraint $\norm{x}_0\leq s$ is not convex; any positive values of $\inner{y-\thr(z)}{z-\thr(z)}$ with $\norm{y}_0\leq s$ measure the extent to which the thresholding operator $\thr$ behaves {\em differently} from a convex projection. By taking a more restrictive constraint on $y$, namely $\norm{y}_0\leq \rho s$ rather than $\norm{y}_0\leq s$, we reduce this measure of concavity; the relative concavity of $\thr$ will be smaller for lower values of $\rho$.

This notion of relative concavity is closely related to the {\em local concavity coefficients} developed in \citet{barber2017gradient} for the purpose of studying projected gradient descent with an arbitrary nonconvex constraint. We will compare the two later on, after presenting our main theorems.

\subsection{Relative concavity and iterative thresholding}
We now examine how the relative concavity of $\thr$ relates
to the convergence behavior of iterative thresholding with a fixed step size. The main message, casted informally, is this:
\begin{quote}
  Given sparsity levels $s$ and $s'=\rho s$, and an $s$-sparse thresholding operator $\thr$, the  condition $\RC_{s,\rho}(\thr)\leq \frac{1}{2\kappa}$ is both necessary and sufficient for restricted optimality to hold relative to sparsity level $s'$.
\end{quote}

\paragraph{Stationary points}
Before giving our formal results, we start with a warm-up---supposing that $x$ is a stationary point of the iterative thresholding algorithm with step size $\eta = \frac{1}{\beta}$, what guarantees can we give about $\loss(x)$? If $\loss$ satisfies $(\alpha,s)$-RSC, then we know that
\[\loss(y)\geq \loss(x) + \inner{y-x}{\nabla\loss(x)} + \frac{\alpha}{2}\norm{x-y}^2_2\]
for any $s$-sparse $y$. Furthermore, writing $z = x - \eta\nabla\loss(x)$, we know that $\thr(z)=x$ since $x$ is a stationary point. Therefore,
\begin{equation}\label{eqn:stationary} \inner{y-x}{\nabla\loss(x)} =- \beta\inner{y-x}{z-x} \geq -\beta\RC_{s,\rho}(\thr)\norm{x-y}^2_2\geq -\frac{\alpha}{2}\norm{x-y}^2_2,\end{equation}
as long as $y$ is $\rho s$-sparse and the relative concavity satisfies $\RC_{s,\rho}(\thr)\leq \frac{1}{2\kappa}$. In other words, this condition on relative concavity is sufficient to ensure that
\[\loss(x)\leq \min_{\norm{y}_0\leq \rho s}\loss(y)\text{ for any stationary point $x$}.\]
Conversely, if $\RC_{s,\rho}(\thr)>\frac{1}{2\kappa}$, Theorem~\ref{thm:lowerbd} below will construct a stationary point $x$ that {\em fails} to satisfy $\loss(x)\leq\min_{\norm{y}_0\leq\rho s}\loss(y)$.

\paragraph{Convergence results}
Next we turn to results for the iterated thresholding algorithm initialized at an arbitrary $s$-sparse point $x_0$ (for example, initialized at zero). Our first theorem accounts for the sufficiency of the condition.
\begin{theorem}\label{thm:upperbd}
Consider any objective function $\loss:\R^d\rightarrow \R$, any sparsity levels $s\geq s'$, and any $s$-sparse thresholding operator $\thr$. Assume the objective function $\loss$ satisfies $(\alpha,s)$-RSC and $(\beta,s)$-RSM. Let $\rho = s'/s$ and $\kappa = \beta/\alpha$, and assume that
\[\RC_{s,\rho}(\thr) < \frac{1}{2\kappa}.\]
Then, for any $s$-sparse $x_0\in\R^d$ and any $s'$-sparse $y\in\R^d$, the iterated thresholding algorithm~\eqref{eqn:iterative_thresh} initialized at $x_0$ and run with fixed step size $\eta=1/\beta$ satisfies
\[\min_{t=1,\dots,T} \loss(x_t) \leq \loss(y) + \left(\frac{1 - 1/\kappa}{1-2\RC_{s,\rho}(\thr)}\right)^T \cdot \frac{\beta}{2} \norm{x_0 - y}^2_2\]
for each $T\geq 1$. The same result holds for the iterative thresholding algorithm with adaptive step size~\eqref{eqn:iterative_thresh_eta_adaptive}.
\end{theorem}
\noindent In other words, the condition $\RC_{s,\rho}(\thr)<\frac{1}{2\kappa}$ guarantees restricted optimality on the class of $\kappa$-conditioned objective functions at sparsity proportion $\rho$. Next, we examine the necessity of the bound on $\RC_{s,\rho}(\thr)$. The following result proves that, if $\RC_{s,\rho}(\thr)>\frac{1}{2\kappa}$, then there exists an objective function $\loss(x)$ on which the restricted optimality guarantee fails, when we run iterative thresholding with fixed step size $\eta = \frac{1}{\beta}$.
\begin{theorem}\label{thm:lowerbd}
Consider any sparsity levels $s\geq s'$, any $s$-sparse thresholding operator $\thr$, and any constants $\beta \geq \alpha >0$. Let $\rho=s'/s$ and $\kappa = \beta/\alpha$, and assume that
\[\RC_{s,\rho}(\thr)>\frac{1}{2\kappa}.\]
Then there exists an objective function $\loss(x)$ that satisfies $(\alpha,s)$-RSC and $(\beta,s)$-RSM, and an $s$-sparse $x_0\in\R^d$ and $s'$-sparse $y\in\R^d$, such that the iterated thresholding algorithm~\eqref{eqn:iterative_thresh} run with step size $\eta=1/\beta$ and initialization point $x_0$ satisfies
\[\lim_{t\rightarrow \infty} \loss(x_t) > \loss(y).\]
\end{theorem}
\noindent This result is proved by constructing an objective function $\loss$ and an $s$-sparse point $x_0$, such that $\loss(x_0)>\loss(y)$, but $x_0$ is a stationary point of the iterated thresholding algorithm, i.e.~by initializing at $x_0$, we obtain $x_t=x_0$ for all $t\geq 1$. This proves that the iterated thresholding algorithm does not satisfy restricted optimality (at the given sparsity levels), since it is trapped at an $s$-sparse point $x_0$ whose objective value is strictly worse than that of the $s'$-sparse point $y$.

\paragraph{Local vs global guarantees}
We have seen that the condition $\RC_{s,\rho}(\thr)<\frac{1}{2\kappa}$ on the relative concavity, is sufficient to ensure a restricted optimality result, without any initialization conditions---that is, this is a global result, rather than a result that holds only in some neighborhood of the optimal solution. We can compare this framework to the local concavity coefficients of \citet{barber2017gradient}, where the convergence guarantee is of a local type.

In the present work, to achieve our convergence result via relative concavity, we require that, for any $z\in\R^d$ and for $x=\thr(z)$, we have
$\inner{y-x}{z-x}\leq \RC_{s,\rho}(\thr)\norm{x-y}^2_2$ for all $\rho s$-sparse  $y$. For a stationary point $x$ of the iterated thresholding algorithm with step size $\eta=\frac{1}{\beta}$, we would have $z = x - \eta\nabla\loss(x)$, and so the requirement above can be rewritten as
\begin{equation}\label{eqn:compare_RC_LCC_1}\inner{y-x}{-\nabla\loss(x)}\leq \beta\cdot \RC_{s,\rho}(\thr)\cdot \norm{x-y}^2_2 \text{ for all $y$ with $\norm{y}_0\leq \rho s$},\end{equation}
 and the term $\beta\RC_{s,\rho}(\thr)$ on the right-hand side is bounded as $\beta\RC_{s,\rho}(\thr)< \beta\cdot\frac{1}{2\kappa} = \frac{\alpha}{2}$ according to the conditions of Theorem~\ref{thm:upperbd}.

In contrast, \citet{barber2017gradient}'s local concavity coefficient framework requires that, for any $z\in\R^d$ and any $y\in\Ccal$, $\inner{y-x}{z-x} \leq\gamma_x(\Ccal) \cdot \norm{z-x}\cdot \norm{y-x}^2_2$ where $x=P_{\Ccal}(z)$ is the projection of $z$ to the constraint set $\Ccal$.
At a stationary point $x$ of projected gradient descent, we have $z = x - \eta\nabla\loss(x)$, and so equivalently,
\begin{equation}\label{eqn:compare_RC_LCC_2}\inner{y-x}{-\nabla\loss(x)} \leq\gamma_x(\Ccal) \cdot \norm{\nabla\loss(x)}\cdot \norm{y-x}^2_2\text{ for all $y\in\Ccal$ }.\end{equation}
\citet{barber2017gradient}'s main results prove convergence to the global minimum over $\Ccal$, as long as the algorithm is initialized in a neighborhood within which the condition $\gamma_x(\Ccal)\norm{\nabla\loss(x)}<\frac{\alpha}{2}$ holds uniformly.\footnote{The norm $\norm{\cdot}$ measuring the magnitude of the gradient $\nabla \loss(x)$ is not necessarily the $\ell_2$ norm---it is typically chosen to be smaller than the $\ell_2$ norm, for instance, the $\ell_{\infty}$ norm in the case of sparse estimation---but this is not relevant to the comparison here.}

Comparing the relative concavity framework~\eqref{eqn:compare_RC_LCC_1} with \citet{barber2017gradient}'s local concavity coefficient framework~\eqref{eqn:compare_RC_LCC_2}, we see that in both settings, $\inner{y-x}{-\nabla\loss(x)}$ is required to strictly less than $\frac{\alpha}{2}\norm{y-x}^2_2$. The difference is that:
\begin{itemize}
\item \citet{barber2017gradient}'s work requires this bound to hold for all $y$ in the constraint set $\Ccal$, but only for $x$ in some neighborhood the global optimum. If the algorithm is initialized in this neighborhood, then global optimality is guaranteed.
\item Our present work requires this bound to hold only for a more restricted set of $y$'s, i.e.~with the restricted sparsity level $\norm{y}_0\leq s'=\rho s$, but for all $x$ in the constraint set of $s$-sparse vectors. Regardless of where the algorithm is initialized, we obtain a restricted optimality guarantee.
\end{itemize}
Overall, by requiring the concavity bound to hold only for a more restricted set of $y$'s, our new result is able to avoid initialization conditions, at the cost of obtaining restricted optimality rather than global optimality as the final guarantee.

\section{Upper and lower bounds on relative concavity}\label{sec:RC}
We have now seen that the relative concavity $\RC_{s,\rho}(\thr)$
fully characterizes the worst-case performance of the thresholding operator $\thr$
in the gradient descent algorithm, with a convergence guarantee in Theorem~\ref{thm:upperbd} and a matching lower bound in Theorem~\ref{thm:lowerbd} (assuming a fixed step size). In this next section, we turn to the question of investigating the relative concavity in greater detail, in order to determine which thresholding operators are most likely to lead to successful optimization. Along the way, we will focus on the following questions:
\begin{itemize}
  \item What is the relative concavity of commonly used thresholding operators, for example, hard thresholding and soft thresholding?
  \item What is the best (i.e.~lowest) possible relative concavity $\RC_{s,\rho}(\thr)$ among all thresholding operators $\thr$, and which thresholding operators are optimal?
\end{itemize}
\noindent Throughout this section, for providing upper and lower bounds on $\RC_{s,\rho}(\thr)$, we will assume without comment that $s,s'\in\{1,\dots,d\}$ are two sparsity levels satisfying $1\leq s'\leq s\leq d$ and $s+s'\leq d$, and we will define $\rho = s'/s$ as usual.

\subsection{Relative concavity of hard and soft thresholding}\label{sec:relativeconcavity_HT_ST}

First, we consider hard thresholding, $\thr = \HT$. The following result
computes the relative concavity for the hard thresholding operator:
\begin{lemma}\label{lem:RC_HT}
The relative concavity of hard thresholding is given by
\[\RC_{s,\rho}(\HT) = \frac{\sqrt{\rho}}{2}\]
for every sparsity proportion $\rho\in(0,1]$.
\end{lemma}
\noindent In particular, with Lemma~\ref{lem:RC_HT}, the condition $\RC_{s,\rho}(\HT)<\frac{1}{2\kappa}$ becomes $\rho<\frac{1}{\kappa^2}$. In light of Theorems~\ref{thm:upperbd} and~\ref{thm:lowerbd}, we see that for iterative hard thresholding algorithm, $\rho<\frac{1}{\kappa^2}$ is necessary and sufficient to guarantee restricted optimality with sparsity level $s$ and $s'$, tightening the condition
obtained in \citet{jain2014iterative} where they prove restricted optimality with the sparsity proportion $\rho = \frac{1}{32\kappa^2}$.

We might wonder whether the highly discontinuous nature of the hard thresholding function might not be ideal---by smoothing out the discontinuity, could we attain better performance? However, we find that any continuous thresholding operator with respect to the Euclidean distance in $\R^d$ is necessarily worse than hard thresholding:
\begin{lemma}\label{lem:RC_contin}
For any continuous map $\thr: \R^d \rightarrow \{x\in\R^d:\norm{x}_0\leq s\}$, its relative concavity satisfies
\[\RC_{s,\rho}(\thr)\geq 1\]
for every sparsity proportion $\rho\in(0,1]$.
\end{lemma}
\noindent In particular, since $\kappa\geq 1$, the condition $\RC_{s,\rho}(\thr)<\frac{1}{2\kappa}$ never holds if $\thr$ is continuous. Comparing to Theorem~\ref{thm:lowerbd}, we see that no continuous operator can guarantee restricted optimality at any sparsity ratio $\rho$, even in the ideal setting where $\loss$ is well-conditioned.
This includes soft thresholding at a fixed sparsity level, i.e., the map $\ST$ that shrinks all entries of $z$ equally until the desired sparsity level is reached:
\[\big(\ST(z)\big)_i = \begin{cases}
z_i - \lambda, & z_i>\lambda, \\
 0, & |z_i|\leq \lambda,\\
z_i + \lambda,&z_i< -\lambda,
\end{cases}\text{ \quad taking $\lambda\geq 0$ to be the smallest value s.t.~$\norm{\ST(z)}_0\leq s$.}\]
In practice, it is much more common to implement soft thresholding at a fixed $\lambda$, rather than at a fixed $s$.
We will discuss the fixed-$\lambda$ formulation of soft thresholding later on, in Section~\ref{sec:softthresh}.

\subsection{Optimal value of relative concavity}
In this section we turn to the question of optimality: what is the optimal value of relative concavity among all thresholding operators at a given sparsity proportion $\rho$? We will establish that
\[\inf_{\thr: \R^d \rightarrow \{x\in\R^d:\norm{x}_0\leq s\}} \RC_{s,\rho}(\thr)=\frac{\rho}{1+\rho}.\]
That is, the lowest relative concavity among all thresholding operators at a given sparsity proportion $\rho$ is exactly $\frac{\rho}{1+\rho}$.
Since this is much smaller than $\frac{\sqrt{\rho}}{2}$ when $\rho$ is small, we see that hard thresholding is suboptimal.

We start with the following lower bound for all thresholding operators:
\begin{lemma}\label{lem:RC_lowerbd}
For any map $\thr: \R^d \rightarrow \{x\in\R^d:\norm{x}_0\leq s\}$ and any sparsity proportion $\rho\in(0,1]$,
the relative concavity is lower-bounded as
\[\RC_{s,\rho}(\thr)\geq \frac{\rho}{1+\rho}.\]
\end{lemma}

\noindent To show that this lower bound is indeed tight, we will consider $\ell_q$ thresholding and establish upper bound for its relative concavity that matches this lower bound with proper choices of $q$. $\ell_q$ thresholding encourage sparsity without exerting too much shrinkage by constraining the $\ell_q$ norm of the vector after thresholding for some $q\in (0,1)$. To be precise, let
\[P_{\ell_q}(z;t) = \arg\min\left\{\norm{x-z}_2 : \norm{x}_q\leq t\right\}\]
denote projection to the $\ell_q$ ball, where $\norm{x}_q =\left(\sum_i |x_i|^q\right)^{1/q}$ is the $\ell_q$ ``norm'' (in fact a nonconvex function since $q<1$). Then define
\[\LQ(z) = P_{\ell_q}(z;t(z)),\text{ where }
t(z) = \sup\left\{t : \norm{P_{\ell_q}(z;t)}_0\leq s\right\}.\]
In words, $\LQ(z)$ projects $z$ to an $\ell_q$ ball whose radius is chosen to be as large as possible while still ensuring $s$-sparsity.\footnote{Note that $P_{\ell_q}(z;t)$ may be non-unique. To be fully precise, we define $\LQ(z)$ by first fixing some map $z\mapsto S$, the possibly non-unique support of its largest $s$ entries, and then defining $t(z)$ and choosing the possibly non-unique projection $P_{\ell_q}(z;t(z))$ in such a way that the nonzero entries in the projection are exactly on this support.} The following result computes the relative concavity for $\ell_q$ thresholding:
\begin{lemma}\label{lem:RC_LQ}
The relative concavity of $\ell_q$ thresholding $\LQ$ is equal to
\[\RC_{s,\rho}(\LQ) = \frac{\frac{\rho}{\min\{1,(\frac{2-q}{q})^2(1-\rho)\}}}{\frac{4q(1-q)}{(2-q)^2}(1+\sqrt{1+(\frac{2-q}{q})^2\frac{\rho}{\min\{1,(\frac{2-q}{q})^2(1-\rho)\}}})}\]
for every sparsity proportion $\rho\in (0,1)$.
In particular, if we choose
\[q = \frac{2(1-\rho)}{3-\rho},\]
then the resulting thresholding operator attains the lowest possible relative concavity,
\[\RC_{s,\rho}(\LQ) = \frac{\rho}{1+\rho},\text{ for } q = \frac{2(1-\rho)}{3-\rho}.\]
In addition, the universal choice $q=2/3$ yields relative concavity equal to,
\[\RC_{s,\rho}(\LQuniv) =  \frac{\frac{\rho}{\min\{1,4(1-\rho)\}}}{\frac{1}{2} + \frac{1}{2}\sqrt{1+\frac{4\rho}{\min\{1,4(1-\rho)\}}}}\leq \frac{\rho}{\min\{1,4(1-\rho)\}},\text{ for all $\rho\in(0,1)$}.\]
\end{lemma}
\noindent Now we provide some explanation for this result. If we are allowed to choose $q$ depend on $\rho$, then the choice $q=\frac{2(1-\rho)}{3-\rho}$ would lead to a relative concavity of $\frac{\rho}{1+\rho}$, which exactly matches the lower bound in Lemma~\ref{lem:RC_lowerbd}. Of course this specific choice of $q$ is chosen for a specific sparsity proportion $\rho$ and might not work well for other values of the sparsity proportion. To avoid this drawback or the need to tune the parameter $q$, one can have the universal choice $q=2/3$. Due to the expression for $\RC_{s,\rho}(\LQuniv)$, we see that $\RC_{s,\rho}(\LQuniv)\approx \rho$ when $\rho$ is small, thus nearly matching the lower bound $\frac{\rho}{1+\rho}$.

In particular, with the optimal value of relative concavity $\RC_{s,\rho}=\frac{\rho}{1+\rho}$, the condition $\RC_{s,\rho}<\frac{1}{2\kappa}$ becomes $\rho<\frac{1}{2\kappa-1}$. In light of  Theorem~\ref{thm:upperbd} and Theorem~\ref{thm:lowerbd}, we see that $\rho<\frac{1}{2\kappa-1}$ is both necessary and sufficient for restricted optimality to hold with sparsity proportion $\rho$.  Compare this with the condition $\rho<\frac{1}{\kappa^2}$ required by hard thresholding, we see that the dependence on condition number is greatly improved!

\subsection{A general class of thresholding operators}\label{sec:generalclass}
Now that we have seen that $\ell_q$ thresholding operators enjoy good properties in terms of relative concavity, we can ask whether there are other thresholding operators of such optimal and near-optimal properties. In this section we address this problem by showing $\ell_q$ thresholding can be characterized as a special case of a larger class of thresholding operators, which all enjoy the same optimality properties in the sense of their relative concavity. Consider any nonincreasing function
\[\sigma: [1,\infty)\rightarrow [0,1],\]
which we call the ``shrinkage function'', which will determine the amount of shrinkage on each entry of a vector $z$ at the thresholding step.
Defining the support $S$ and thresholding level $\tau=\max_{i\not\in S}|z_i|$ as before, we then define the thresholding operator $\thrsig$ as
\[\big(\thrsig(z)\big)_i = \begin{cases}
z_i- \tau\sigma\big(|z_i|/\tau\big),&i\in S,\\
0,&i\not\in S.\end{cases}\]
In other words, for entry $i\in S$, $\sigma\big(|z_i|/\tau\big)$ determines the {\em relative} amount of shrinkage on this entry. The intuitive meaning of $\sigma$ is illustrated in Figure~\ref{fig:explaindef}. (If $\tau = 0$, i.e.~$z$ is already $s$-sparse, then we would simply take $\thrsig(z)=z$; we will ignore this case from this point on.)

\begin{figure}[t]
\centering
\includegraphics[width=0.7\textwidth]{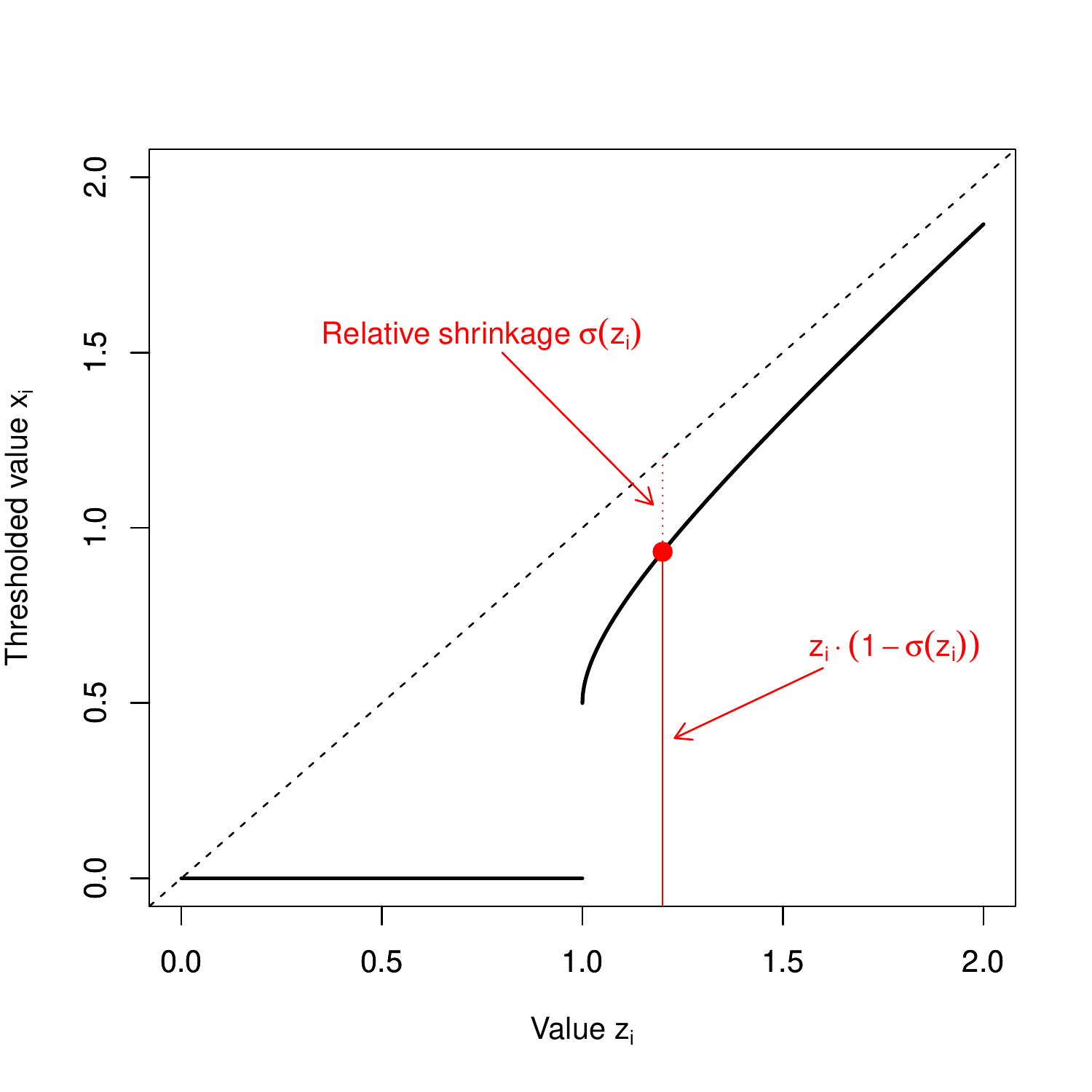}
\caption{An illustration of the definition of the relative shrinkage function $\sigma$ (for simplicity we use thresholding level $\tau=1$ in this illustration). Here we use the relative shrinkage function $\sigma(t) = \frac{t-\sqrt{t^2-1}}{2}$, corresponding to the reciprocal thresholding operator defined in Section~\ref{sec:RT}.}
\label{fig:explaindef}
\end{figure}

Note that since $\sigma$ is nondecreasing, the maximum shrinkage occurs when $|z_i|=\tau$ exactly; the amount of shrinkage in this setting is governed by $\sigma(1)$.

We can now examine the relationship of the choice of $\sigma$ to the relative concavity:
\begin{lemma}\label{lem:unifying_RC}
For any nonincreasing shrinkage function $\sigma:[1,\infty)\rightarrow[0,1]$ such that $0<\sigma(1)<1$ and
\begin{equation}\label{eqn:sigmono}
t\mapsto \sigma(t)(t-\sigma(t))\text{ is nondecreasing over $t\geq 1$},\end{equation}
the thresholding operator $\thrsig$ has relative concavity
\[\RC_{s,\rho}(\thrsig) = \frac{\frac{\rho}{\min\{1,(1-\rho)/\sigma(1)^2\}}}{2\sigma(1)\big(1-\sigma(1)\big)\left(1 + \sqrt{1 + \frac{\rho/\sigma(1)^2}{\min\{1,(1-\rho)/\sigma(1)^2\}}}\right)}.\]
In particular, the resulting operator attains the lowest possible relative concavity,
\[\RC_{s,\rho}(\thrsig) = \frac{\rho}{1+\rho},\]
if and only if $\sigma(1) = \frac{1-\rho}{2}$.
If instead we take a universal shrinkage level $\sigma(1)= \frac{1}{2}$, then the relative concavity is given by
\[\RC_{s,\rho}(\thrsig) =
\frac{\frac{\rho}{\min\{1,4(1-\rho)\}}}{\frac{1}{2} + \frac{1}{2}\sqrt{1+\frac{4\rho}{\min\{1,4(1-\rho)\}}}}
\leq \frac{\rho}{\min\{1,4(1-\rho)\}}.\]
\end{lemma}
Examining the definition of this general family of thresholding operators, we can see that $\ell_q$ thresholding corresponds to setting
\[\sigma(t)= t-\left(\text{the larger-magnitude root $x$ of the equation $t = x + \frac{q(2-2q)^{1-q}/(2-q)^{2-q}}{x^{1-q}}$}\right),\]
for which we have $\sigma(1) = \frac{q}{2-q}$ and which satisfies~\eqref{eqn:sigmono}. We also have that $\sigma(1) = \frac{1}{2}$ corresponds to the ``universal'' choices $q=2/3$, and $\sigma(1) = \frac{1-\rho}{2}$ (the optimal value) corresponds to the $\rho$-specific choices $q = \frac{2(1-\rho)}{3-\rho}$. As a consequence, the previous result for $\ell_q$ thresholding, Lemma~\ref{lem:RC_LQ}, is simply special case of this more general lemma. 

On the other hand, the hard thresholding operator $\HT$ can be obtained by setting $\sigma(t)=0$ for all $t\in[1,\infty)$, but this does not satisfy the assumption $\sigma(1)>0$ required in the lemma. However, if we informally consider fixing $\rho>0$ and taking a limit $\sigma(1)\rightarrow 0$ in the upper bound in the lemma, we see
\[\lim_{\sigma(1)\rightarrow 0}\frac{\frac{\rho}{\min\{1,\frac{1-\rho}{\sigma(1)^2}\}}}{2\sigma(1)(1-\sigma(1))\left(1+\sqrt{1+\frac{\rho}{\sigma(1)^2\min\{1,\frac{1-\rho}{\sigma(1)^2}\}}}\right)}= \frac{\sqrt{\rho}}{2},\]
obtaining the relative concavity of hard thresholding calculated earlier.

\subsection{Reciprocal thresholding and minimal shrinkage}\label{sec:RT}
Practically, for two thresholding operators with the same restricted optimality guarantees, i.e. with the exact same value of relative concavity, we may favor the one that exerts smaller amount of shrinkage. Thus it makes sense to ask among the general class of thresholding operators defined in Section~\ref{sec:generalclass}, which operators exert the minimal amount of shrinkage? Consider all operators of the form $\thrsig$, with some fixed value of $\sigma(1)\in (0,1/2]$. For any $\sigma$ satisfying the assumption~\eqref{eqn:sigmono}, for all $t\geq 1$ we have
\[\sigma(t)(t-\sigma(t)) \geq \sigma(1)(1-\sigma(1)).\]
For convenience, we reparametrize this equation by setting $c = 1-2\sigma(1) \in[0,1)$, and so we are considering all nonincreasing functions $\sigma:[1,\infty)\rightarrow[0,1]$ that satisfy $\sigma(1) = \frac{1-c}{2}$ and
\[\sigma(t)(t-\sigma(t)) \geq \sigma(1)(1-\sigma(1)) = \frac{1-c^2}{4}.\]
Thus, we must have
\begin{equation}\label{eqn:sigmalowerbd}
\sigma(t) \geq \frac{t - \sqrt{t^2 - (1-c^2)}}{2}
\end{equation}
for all $t\geq 1$.

This motivates a new family of thresholding operators, {\em reciprocal thresholding with parameter $c$}, which is designed to make the inequality~\eqref{eqn:sigmalowerbd} an equality. To be specific, we define reciprocal thresholding with parameter $c$ to be
\[\RTc=\thrsig\text{ with shrinkage function }
\sigma(t)=\frac{t-\sqrt{t^2-(1-c^2)}}{2}.\]
To apply this operator to some vector $z\in\R^d$, we first let $S\subset\{1,\dots,d\}$ be the indices of the largest $s$ entries of $z$ (with our usual caveat about needing to establish some rule for breaking ties) and let $\tau = \max_{i\not\in S}|z_i|$ be the magnitude of the $(s+1)$-st largest entry of $z$. Then $
\RTc(z)$ operates entry-wise as follows:
\begin{equation}\label{eqn:define_RTc}\big(\RTc(z)\big)_i = \begin{cases}
\textnormal{sign}(z_i)\cdot \left(\frac{1}{2}|z_i| + \frac{1}{2}\sqrt{|z_i|^2 - \tau^2(1-c^2)}\right),&\text{ if $i\in S$,}\\
0,&\text{ if $i\not\in S$.}
\end{cases}\end{equation}
Here the  thresholded value  $\big(\RTc(z)\big)_i$ is equal to the larger-magnitude root $t$ of the equation
\begin{equation}\label{eqn:explain_RTc} z_i = t + \frac{\tau^2 \cdot \frac{1-c^2}{4}}{t},\end{equation}
hence the name ``reciprocal thresholding''.

As before, to avoid the need for selecting $c$ adaptively, we might want to consider some fixed choices. At one extreme, taking $c=1$ yields $\RTone = \HT$, the hard thresholding operator. At the other extreme, taking $c=0$ defines the ``universal'' {\em reciprocal thresholding} operator:
\[\RT=\RTzero.\]
For any $z\in\R^d$, $\RT$ operate entry-wise as:
\begin{equation}\label{eqn:define_RT}\big(\RT(z)\big)_i = \begin{cases}
\textnormal{sign}(z_i)\cdot \left(\frac{1}{2}|z_i| + \frac{1}{2}\sqrt{|z_i|^2 - \tau^2}\right),&\text{ if $i\in S$,}\\
0,&\text{ if $i\not\in S$.}
\end{cases}\end{equation}

The following lemma calculates the relative concavity of $\RTc$ and $\RT$ as a direct consequence of Lemma~\ref{lem:unifying_RC}.
\begin{lemma}\label{lem:RC_RT}
For any sparsity proportion $\rho\in(0,1]$, the thresholding operator $\RTc$ with parameter $c=\rho$ has relative concavity equal to
\[\RC_{s,\rho}(\thr^{\textnormal{RT},\rho})= \frac{\rho}{1+\rho}.\]
The reciprocal thresholding operator $\RT$ has relative concavity equal to
\[\RC_{s,\rho}(\RT)=
\frac{\frac{\rho}{\min\{1,4(1-\rho)\}}}{\frac{1}{2} + \frac{1}{2}\sqrt{1+\frac{4\rho}{\min\{1,4(1-\rho)\}}}} \leq \frac{\rho}{\min\{1,4(1-\rho)\}}\]
for every sparsity proportion $\rho\in(0,1)$.
\end{lemma}
\noindent Thus, $\RTc$ with $c=\rho$ is exactly optimal among all thresholding operators relative to the sparsity proportion $\rho$ (as is $\LQ$ with $q=\frac{2(1-\rho)}{3-\rho}$), while $\RT$ is near optimal when $\rho$ is small (as is $\LQuniv$).

\subsection{An illustrative comparison}
Through the development in this section, we see that there are three important benchmarks for relative concavity: the bound $\sqrt{\rho}/2$ attained by hard thresholding $\HT$, the bound $\frac{\frac{\rho}{\min\{1,4(1-\rho)\}}}{\frac{1}{2} + \frac{1}{2}\sqrt{1+\frac{4\rho}{\min\{1,4(1-\rho)\}}}}$ attained by reciprocal thresholding $\RT$ and $\ell_{2/3}$ thresholding $\LQuniv$, and the optimal value $\frac{\rho}{1+\rho}$. In this section we provide a comparison between these three values.

The left-hand plot of Figure~\ref{fig:compare_bounds} displays the three values of relative concavity as functions of the sparsity proportion $\rho$. We see that at small values $\rho\approx 0$, the relative concavity of reciprocal thresholding and $\ell_{2/3}$ thresholding is nearly identical to the optimal bound $\frac{\rho}{1+\rho}$, and is substantially better than the relative concavity for hard thresholding, given by $\sqrt{\rho}/2$. At larger values of $\rho$, the relative concavity for hard thresholding is instead lower.

\begin{figure}[t]
\centering
\includegraphics[width=\textwidth]{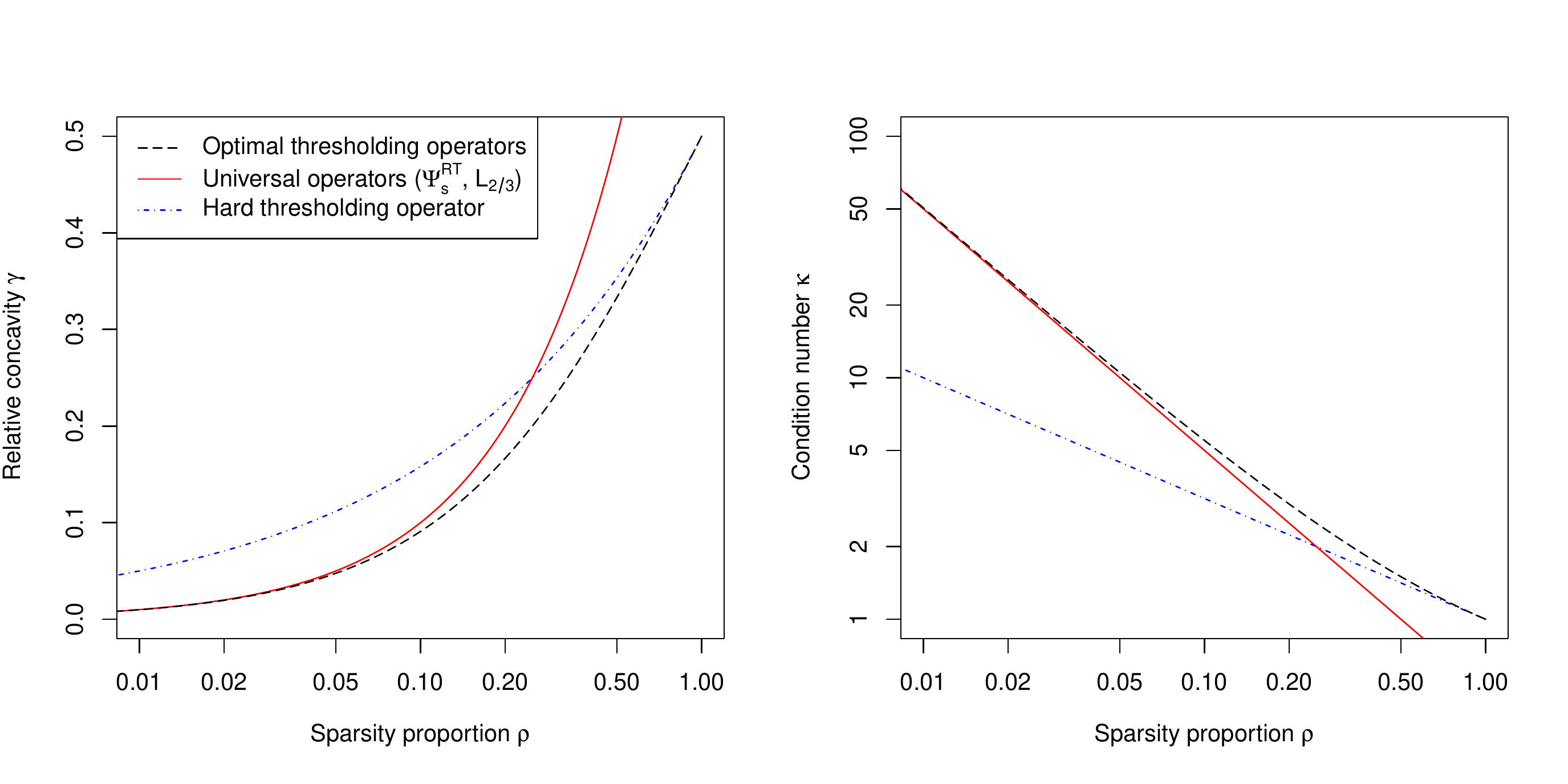}
\caption{A comparison of three values of relative concavity: the optimal relative concavity (attained by, for instance, $\RTc$ with $c=\rho$, or by $\LQ$ with $q=\frac{2(1-\rho)}{3-\rho}$); the relative concavity obtained by the ``universal'' operators, including $\RT$ and by $\ell_{2/3}$ thresholding; and the relative concavity of hard thresholding. The left plot shows the relative concavity as a function of the sparsity proportion $\rho$. The right plot shows the largest possible condition number $\kappa$ of the objective function $\loss$ for which a restricted optimality guarantee can be attained (Theorems~\ref{thm:upperbd} and~\ref{thm:lowerbd} show that $\RC_{s,\rho}(\thr) \leq \frac{1}{2\kappa}$ is necessary and sufficient for a restricted optimality guarantee).}
\label{fig:compare_bounds}
\end{figure}

To view this comparison in another light, given any fixed thresholding operator $\thr$ with certain relative concavity, and given an objective function $\loss$ with condition number $\kappa$, for what sparsity ratio $\rho = s'/s$ is the iterative thresholding algorithm guaranteed to achieve restricted optimality? Using the condition $\RC_{s,\rho}(\thr)\leq \frac{1}{2\kappa}$, for each relative concavity $\RC_{s,\rho}$ we can solve for the largest possible $\kappa$ for which restricted optimality is assured, as a function of $\rho$. 

 This is illustrated in the right-hand plot of Figure~\ref{fig:compare_bounds}, where we see that the reciprocal thresholding operator $\RT$ and the $\ell_{2/3}$ thresholding operator achieve a nearly-optimal sparsity ratio $\rho$ when the condition number $\kappa$ is large and $\rho$ is correspondingly close to zero, while hard thresholding $\HT$ is closer to optimal for $\kappa$ and $\rho$ close to $1$. Thus, we can conclude that reciprocal thresholding and $\ell_{2/3}$ thresholding offer stronger theoretical guarantees when $\kappa>2$, while hard thresholding may be better for very well-conditioned problems where $1\leq \kappa <2$. (Empirically, we have observed that it is often the case that the three perform nearly identically in ``generic'' problems, and only show substantial differences in problems constructed to mimic our lower bound result, Theorem~\ref{thm:lowerbd}, for example, in linear regression problems where a small subset of the features are generated to have covariance structure similar to the construction in Theorem~\ref{thm:lowerbd}.)

\subsection{A closer look at soft thresholding}\label{sec:softthresh}
In many applications, it is common to use a sparsity-inducing penalty rather than an explicit sparsity constraint. For example,
we may solve
\[\widehat{x} = \arg\min_{x\in\R^d}\big\{\loss(x) + \lambda \norm{x}_1\big\},\]
which is known as the Lasso~\citep{tibshirani1996regression} in the context of regression problems.
More generally, we can consider 
\begin{equation}\label{eqn:xhatlambda}\widehat{x}_{\lambda} = \arg\min_{x\in\R^d}\big\{\loss(x) + \lambda \reg(x)\big\},\end{equation}
where $\reg:\R^d\rightarrow\R$ is any proper convex function acting as a regularizer.
This class of problems can be solved iteratively with a proximal gradient method, 
\begin{equation}\label{eqn:iterative_prox}
x_t = \textnormal{Prox}_{\lambda\eta\reg}\big(x_{t-1} - \eta \nabla \loss(x_{t-1})\big),
\end{equation}
where for any $t\geq 0$, the proximal map is defined as
\[\textnormal{Prox}_{t\reg}(z) = \arg\min_{x\in\R^d}\left\{\frac{1}{2}\norm{x-z}^2_2 + t\reg(x)\right\}.\]
Note that convexity of $\reg(x)$ ensures continuity of the proximal map. More properties of the proximal map and the proximal method can be found in \citep{parikh2014proximal}.
Examining the iterations of proximal gradient descent~\eqref{eqn:iterative_prox}, we see that it is very similar
to the iterative thresholding method~\eqref{eqn:iterative_thresh} for a fixed sparsity level $s$ (using some particular
thresholding operator $\thr$); we simply replace the thresholding operator $\thr$ with the proximal map $\textnormal{Prox}_{\lambda\eta\reg}$.

In particular, if we consider $\reg(x) =\norm{x}_1$, the resulting proximal map is known as ``soft thresholding'', and can be 
computed with elementwise shrinkage:
\[\big(\textnormal{Prox}_{t\norm{\cdot}_1}(z) \big)_i =\begin{cases}
z_i - t, & z_i > t, \\
0, & |z_i|\leq t,\\
z_i + t, & z_i< -t.\end{cases}\]
Now, recall that in Section~\ref{sec:relativeconcavity_HT_ST}, we considered a ``soft thresholding'' operator at a {\em fixed} sparsity
level $s$, which we can now rewrite as
\[\ST(z) = \textnormal{Prox}_{t\norm{\cdot}_1}(z) \textnormal{\quad where $t\geq 0$ is the smallest value s.t.~$\norm{ \textnormal{Prox}_{t\norm{\cdot}_1}(z) }_0 \leq s$.}\]
We might ask whether the suboptimal worst-case performance of iterative thresholding with the operator $\ST$, as established
by Lemma~\ref{lem:RC_contin} and Theorem~\ref{thm:lowerbd}, 
is due to the unusual definition of $\ST$, using a fixed sparsity level $s$, rather than the usual form of soft thresholding where
we would iterate~\eqref{eqn:iterative_prox} at a fixed value of $\lambda$ in order to minimize  $\loss(x) + \lambda \norm{x}_1$.

In fact, we will now see that this is not the case---even if we use a fixed $\lambda$ rather than a fixed sparsity level $s$,
we can still find worst-case examples where restricted optimality is not achieved.
\begin{theorem}\label{thm:prox_worstcase}
Let $d\geq 2$, let $\beta\geq \alpha >0$, and let $\reg:\R^d\rightarrow\R$ be a proper convex function 
that satisfies the following assumptions:
\begin{align}
\label{eqn:prox_assump1}&\text{For any $z\in \R^d$ and any $t'>t\geq 0$, if $\norm{\textnormal{Prox}_{t\reg}(z)}_{0}<d$ then $\norm{\textnormal{Prox}_{t'\reg}(z) }_{0}<d$.}\\
\label{eqn:prox_assump2}&\text{There exist $v,w \in \R^d$ that are both dense, i.e., $\norm{v}_0=\norm{w}_0=d$,  with $w \in \partial \reg(v)$.}
\end{align}Then there exists an objective function $\loss(x)$ that satisfies $(\alpha,d)$-RSC and $(\beta,d)$-RSM, and a 1-sparse vector $y\in\R^d$, such that
defining $\widehat{x}_\lambda$ as in~\eqref{eqn:xhatlambda}, 
\[\text{For all $\lambda\geq 0$, either $\norm{\widehat{x}_\lambda}_0=d$ or $\loss(\widehat{x}_\lambda) > \loss(y)$.}\]
\end{theorem}
\noindent In other words, this result means there is no value of $\lambda$ that produces a solution that is both sparse (at any sparsity level $<d$) and 
has an objective function value at least as good as the best 1-sparse solution $y$.

We remark that 
our conditions~\eqref{eqn:prox_assump1} and~\eqref{eqn:prox_assump2} on the regularizer $\reg$ are satisfied by many common regularizers---for
 example, the $\ell_1$ norm (Lasso), any $\ell_p$ norm for $1\leq p\leq \infty$, the elastic net (a combination
of the $\ell_1$ and $\ell_2$ norms), the weighted $\ell_1$ norm, and many others.
To help interpret the first condition~\eqref{eqn:prox_assump1}, this essentially requires that a sparse solution $\widehat{x}_\lambda$ will stay sparse if we increase
the penalty parameter $\lambda$, as we would expect for any sparsity-promoting regularizer.

This theorem implies that, just as continuous thresholding operators $\thr$ at a fixed level $s$ 
can fail to attain restricted optimality in a worst-case scenario, the same holds for regularization with convex penalties (such as soft
thresholding
with the $\ell_1$ norm). An open question remains here, namely, is there a measure in the style of relative concavity,
which can characterize the worst-case performance of penalty functions $\reg$ (covering both convex and nonconvex penalty functions,
just as relative concavity treats both continuous and non-continuous thresholding operators)?

\section{Iterative thresholding for low-rank matrices}\label{sec:matrix}
We next extend our analysis of iterative thresholding methods to the setting of a low-rank constraint. In fact, our results carry over fully into this setting. Given a rank constraint, $\rank(X)\leq s$, the hard thresholding operator is defined as
\[\mHT: X\mapsto U\cdot\textnormal{diag}(\HT(d))\cdot V^\top,\]
where $X=U\cdot\textnormal{diag}(d)\cdot V^\top$ is the singular value decomposition of $X$.\footnote{In the case of repeated singular values, the singular value decomposition will not be unique, and we assume that we have some mechanism for specifying a specific singular value decomposition. This is analogous to the sparse vector problem, where if the $s$th largest entry in $z$ is not unique, we need to assume some mechanism for breaking ties and choosing the support of the thresholded vector.} That is, hard thresholding is performed on the singular values of the matrix $X$, rather than on its entries. Of course, we can extend this to any thresholding operator---given any $\thr:\R^{\min\{n,m\}}\rightarrow\{x\in\R^{\min\{n,m\}}:\norm{x}_0\leq s\}$, we can ``lift'' this thresholding operator to the matrix setting by defining
\begin{equation}\label{eqn:mthr_lift}\mthr:X\mapsto U\cdot \textnormal{diag}(\thr(d))\cdot V^\top.\end{equation}
Of course, its possible to construct a rank-$s$ thresholding operator $\mthr$ that is not of the form given in~\eqref{eqn:mthr_lift}, for example, if $\mthr$ does not preserve the left and right singular vectors of $Z$.

We next extend our convergence results, Theorems~\ref{thm:upperbd} and~\ref{thm:lowerbd}, to the low-rank setting. In order to do so, we need to define the matrix version of relative concavity---this definition is analogous to the vector case, with rank constraints in place of sparsity constraints:
\[\mRC_{s,\rho}(\mthr) = \sup\left\{\frac{\inner{Y-\mthr(Z)}{Z-\mthr(Z)}}{\fronorm{Y-\mthr(Z)}^2} \ : Y,Z\in\R^{n\times m}, \rank(Y)\leq \rho s, Y\neq \mthr(Z)\right\}.\]
As for the vector case, relative concavity is necessary and sufficient for guaranteeing restricted optimality---in fact, the proofs of these are completely identical to the vector case. For completeness, we state the results here, for the matrix version of the iterated thresholding algorithm:
\begin{equation}\label{eqn:iterative_thresh_matrix}
X_t = \mthr\big(X_{t-1} - \eta_t \nabla \loss(X_{t-1})\big),
\end{equation}
with either fixed step size $\eta_t=1/\beta$ or adaptive step size defined as in~\eqref{eqn:iterative_thresh_eta_adaptive}.
\begin{theorem}\label{thm:matrix_upperbd}
Consider any objective function $\loss:\R^{n\times m}\rightarrow \R$, any ranks $s\geq s'$, and any rank-$s$ thresholding operator $\mthr$. Assume the objective function $\loss$ satisfies $(\alpha,s)$-RSC and $(\beta,s)$-RSM relative to the rank constraint.\footnote{In the low-rank setting, the RSC and RSM conditions are defined with rank in place of sparsity---specifically, we are assuming that $\frac{\alpha}{2}\fronorm{X-Y}^2 \leq \loss(Y)- \loss(X) - \inner{\nabla \loss(X)}{Y-X} \leq \frac{\beta}{2}\fronorm{X-Y}^2$ whenever $\rank(X)\leq s,\rank(Y)\leq s$.}
Let $\rho = s'/s$ and $\kappa = \beta/\alpha$, and assume that
$\mRC_{s,\rho}(\mthr) < \frac{1}{2\kappa}$.
Then, for any $X_0,Y\in\R^{n\times m}$ with $\rank(X_0)\leq s$ and $\rank(Y)\leq s'$, the iterated thresholding algorithm~\eqref{eqn:iterative_thresh_matrix} run with step size $\eta=1/\beta$ and initialization point $X_0$ satisfies
\[\min_{t=1,\dots,T} \loss(X_t) \leq \loss(Y) + \left(\frac{1 - 1/\kappa}{1-2\mRC_{s,\rho}(\mthr)}\right)^T \cdot \frac{\beta}{2} \fronorm{X_0 - Y}^2\]
for each $T\geq 1$.
\end{theorem}
\begin{theorem}\label{thm:matrix_lowerbd}
Consider any ranks $s\geq s'$, any rank-$s$ thresholding operator $\mthr$, and any constants $\beta \geq \alpha >0$. Let $\rho=s'/s$ and $\kappa = \beta/\alpha$, and assume that
$\mRC_{s,\rho}(\mthr)>\frac{1}{2\kappa}$.
Then there exists an objective function $\loss(X)$ that satisfies $(\alpha,s)$-RSC and $(\beta,s)$-RSM relative to the rank constraint, and matrices $X_0,Y\in\R^{n\times m}$ with $\rank(X_0)\leq s$ and $\rank(Y)\leq s'$, such that the iterated thresholding algorithm~\eqref{eqn:iterative_thresh_matrix} run with step size $\eta=1/\beta$ and initialization point $X_0$ satisfies
\[\lim_{t\rightarrow \infty} \loss(X_t) > \loss(Y).\]
\end{theorem}
\noindent In other words, just as for the sparse optimization problem, the relationship between relative concavity and condition number gives a necessary and sufficient condition for guaranteed convergence. We note that these results apply to {\em any} rank-$s$ thresholding operator $\mthr$, whether or not it can be constructed by ``lifting'' a $s$-sparse thresholding operator as in~\eqref{eqn:mthr_lift}.

Next, how can we calculate relative concavity of a thresholding operator in the matrix setting? For simplicity, from this point on we assume that we are working with ranks $s\geq s'\geq 1$ with $s+s'\leq \min\{n,m\}$. For this question, we will again see that results from the sparse setting transfer to the low-rank setting. First, we have the same lower bound uniformly over all operators:
\begin{lemma}\label{lem:RC_lowerbd_matrix}
For any map $\mthr: \R^{n\times m} \rightarrow \{X\in\R^{n\times m}:\rank(X)\leq s\}$ and any sparsity proportion $\rho\in(0,1]$,
the relative concavity is lower-bounded as
\[\mRC_{s,\rho}(\mthr)\geq \frac{\rho}{1+\rho}.\]
\end{lemma}
\noindent Furthermore, if we restrict our attention to ``lifted'' thresholding operators of the form~\eqref{eqn:mthr_lift}, the relative concavity of $\thr$ is inherited by the lifted operator $\mthr$---as long as we restrict ourselves to $s$-sparse thresholding operators $\thr$ that satisfy a natural sign condition:
\begin{equation}\label{eqn:thr_sign}
\text{For any $z\in\R^d$ and any $a\in\{\pm 1\}^d$, $\thr\big(\textnormal{diag}(a)\cdot z\big) = \textnormal{diag}(a)\cdot \thr(z)$.}
\end{equation}
This effectively means that $\thr(z)$ preserves the signs of $z$, but the signs of $z$ do not affect the amount of shrinkage in the thresholded vector $\thr(z)$. For example, this requires that $\thr(-z) = - \thr(z)$. Under this assumption, the relative concavity of $\thr$ carries over into the matrix setting.
\begin{lemma}\label{lem:matrix_RC}
Let $\thr$ be a $s$-sparse thresholding operator satisfying the sign condition~\eqref{eqn:thr_sign}, and let $\mthr$ be the lifted thresholding operator defined in~\eqref{eqn:mthr_lift}. Then for every sparsity proportion $\rho\in(0,1]$,
\[\mRC_{s,\rho}(\mthr) = \RC_{s,\rho}(\thr).\]
\end{lemma}
\noindent It is obvious that all the thresholding operators we have considered satisfy the sign condition~\eqref{eqn:thr_sign}. Thus, all the results of relative concavity that we have proved in the sparse setting, carry over directly to the low-rank setting. In particular, as for the sparse setting, the hard thresholding operator has relative concavity
\[\mRC_{s,\rho}(\mHT) = \frac{\sqrt{\rho}}{2},\]
while any thresholding operator $\mthrsig$ constructed with some shrinkage function $\sigma$ satisfying $\sigma(1)=1/2$  and the conditions of Lemma~\ref{lem:unifying_RC}, such as the reciprocal thresholding operator, $\mRT$, or $\ell_q$ thresholding with $q=2/3$, $\mthr^{\ell_{2/3}}$, satisfy
\[\mRC_{s,\rho}(\mthrsig)=\mRC_{s,\rho}(\mthr^{\ell_{2/3}})=\mRC_{s,\rho}(\mRT) \leq \frac{\rho}{\min\{1,4(1-\rho)\}}.\]
If the desired rank proportion $\rho=s'/s$ is fixed in advance, then as before, choosing reciprocal thresholding with parameter $c=\rho$, or $\ell_q$ thresholding with $q = \frac{2(1-\rho)}{3-\rho}$, we again obtain the optimal relative concavity of $\frac{\rho}{1+\rho}$. As before, we can conclude that reciprocal thresholding and $\ell_{2/3}$ each offer lower relative concavity than hard thresholding whenever $\rho$ is small---and, correspondingly, are a safer choice for objective functions $\loss$ whose condition number is not close to $1$.

\section{Sparse linear regression}\label{sec:linear_regression}
Now that we have discussed the deterministic optimization setting in depth, it is natural to ask what is the implication of these guarantee for a statistically random setting. In this section, we apply our developed machinery to the concrete statistical setting of sparse linear regression. We work with the Gaussian linear model
\begin{equation}\label{eqn:linearmodel}
y=X\theta_0+z
\end{equation}
where $X\in \R^{n\times p}$ is a fixed design matrix, $\theta_0\in \R^p$ is the true coefficient vector assumed to be fixed and $s_0$-sparse, and $z\sim N(0,\sigma^2\mathbf{I}_n)$ is the noise vector, with fixed unknown noise level $\sigma^2>0$. In this section we will mainly be interested in prediction error, i.e. how well we can estimate the true mean vector $X\theta_0$. One way of capturing the conditioning of the design matrix is by the following definition: at some given sparsity level $s$, we define a set of design matrices $\Xcal(\alpha,\beta,s)$ as
\begin{equation}\label{eqn:Xset}
\Xcal(\alpha,\beta,s)=\left\{X\in\R^{n\times p}:\textnormal{ the map $\theta\mapsto \theta^\top \left(\frac{X^\top X}{2n}\right)\theta$ satisfies $(\alpha,s)$-RSC and $(\beta,s)$-RSM}\right\}.
\end{equation} As usual, we will be interested in the condition number $\kappa = \beta/\alpha$.  A similar definition is the {\em restricted eigenvalue} condition on the design matrix $X$, which constrains $X$ to the following set
\begin{multline}\label{eqn:Xset_RE}
\Xcal_{\textnormal{RE}}(\kappa,s_0)=\bigg\{X\in\R^{n\times p}:\textnormal{ $\max_{j=1,\dots,p}\frac{\norm{X_j}_2}{\sqrt{n}}\leq 1$, and $\theta^\top \left(\frac{X^\top X}{2n}\right)\theta\geq \frac{1}{2\kappa}\norm{\theta}^2_2$}\\\text{ for all $\theta\in\R^d$ with $\norm{\theta}_1\leq 4\max_{|S|=s_0}\norm{\theta_S}_1 $}\bigg\}.
\end{multline} To gain some intuition for when these conditions may hold, for a design matrix $X$ whose rows are i.i.d.~draws from a normal distribution $N(0,\Sigma)$, \citet[Theorem 1]{raskutti2010restricted} show that the population-level eigenvalues of the covariance $\Sigma$ are approximately preserved in the design matrix, at any sparsity level $s\ll \frac{n}{\log(p)}$.

\paragraph{Computational lower bound}
In terms of prediction error, the optimal method, $\ell_0$ constrained least squares method, is not computable. Thus from the lower bound side, it is of interest to ask what is the lowest prediction error achievable in the class of computational feasible estimator. Recently, \citet{zhang2014lower} provide a partial answer to this question, restricting to the class of $s_0$ sparse estimator. Their main result (see Theorem $1$ in \citet{zhang2014lower}) states the following (informally):
\begin{quote}
  Under the assumption that $\emph{NP}\nsubseteq \emph{P}\setminus \emph{poly}$, for any $\delta\in (0,1)$, under some assumption on $n,d,s_0$ and for any $\kappa$ in a wide range, there exists a design matrix $X\in \Xcal_{\textnormal{RE}}(\kappa,s_0)$ such that for any computational efficient methods, the maximum prediction error (over all $s_0$ sparse $\theta_0$) is lower bounded by (up to some constant) $\kappa \cdot \frac{\sigma^2 s_0^{1-\delta}\log(d)}{n}$.
\end{quote}
Thus if we restrict ourselves to all computationally feasible $s_0$ sparse estimator, then the best achievable squared prediction error is of order $\kappa \cdot \frac{\sigma^2 s_0\log(d)}{n}$.

\paragraph{Upper bounds for iterative thresholding methods}
In this section we establish prediction error bounds for iterative thresholding algorithm. First we provide some intuition on how to connect restricted optimality guarantee with statistical performance. It is well known that the global optimum of $\ell_0$-constrained least squared loss, i.e.
\[\widehat{\theta}\in\arg\min_{\norm{\theta}_0\leq s_0}\norm{y-X\theta}^2_2,\]
achieves a squared prediction error scaling as $\frac{\sigma^2s_0\log(d)}{n}$. For iterative thresholding algorithms, since we only have restricted optimality rather than global optimality, we are forced to work over a constraint at a larger sparsity $s\geq s_0$ to guarantee $\norm{y-X
\widehat\theta}^2_2\leq\min_{\norm{\theta}_0\leq s_0}\norm{y-X\theta}^2_2$.
The statistical price one has to pay for this computational strategy is the inflation in noise level corresponding to the inflation in sparsity---that is, we have error on $s$ many nonzero coefficients, rather than $s_0$ many---so the final upper bound for prediction error would scale as $\frac{\sigma^2s\log(d)}{n}$ instead of $\frac{\sigma^2s_0\log(d)}{n}$, where $s$ is chosen to be the smallest sparsity level that guaratees restricted optimality relative to the lower sparsity level $s'=s_0$. Now recall from Section~\ref{sec:RC} that, while hard thresholding offers restricted optimality guarantees at sparsity levels $s\sim \kappa^2s_0$, the optimal and near-optimal thresholding operators (for example reciprocal thresholding and $\ell_{2/3}$ thresholding) improves this scaling to $s\sim \kappa s_0$. This allows us to improve the upper bound for squared prediction error from scaling as $\kappa^2$ to $\kappa$, when we switch our method from iterative hard thresholding, to iterative thresholding with an operator $\thr$ that enjoys a lower relative concavity. Indeed in \citet{jain2014iterative}, it is shown that iterative hard thresholding achieves a prediction error upper bounded by $\kappa^2\cdot\frac{\sigma^2s_0\log(d)}{n}$. In view of our lower bound result Theorem~\ref{thm:lowerbd}, which states that the restircted optimality guarantee is tight, we postulate that the corresponding prediction error bound is also tight for iterative hard thresholding method.

Now we formulate this rigorously. Consider the iterative thresholding algorithm with some thresholding operator $\thr$ applied to the objective function $\loss(\theta) = \frac{1}{2n}\norm{y-X\theta}^2_2$, whose iteration takes the form
\begin{equation}\label{eqn:RT_linearreg}\widehat\theta_t = \thr\big(\widehat\theta_{t-1} + \eta_t \cdot\frac{1}{n} X^\top (y - X\widehat\theta_{t-1})\big).\end{equation}
As usual, for the step size we may choose $\eta_t = 1/\beta$ if $\beta$ is known, or we may choose $\eta_t$ adaptively as in~\eqref{eqn:iterative_thresh_eta_adaptive}. We will work with any thresholding operator $\thr$ satisfying
\begin{equation}\label{eqn:RC_rho_for_regr}\RC_{s,\rho}(\thr)\leq \rho\text{ for all $\rho\in(0,1/2)$}.\end{equation}
From Section~\ref{sec:RC}, we see that on the one hand, this condition rules out hard thresholding and any continuous thresholding operator; on the other hand, it is satisfied by the reciprocal thresholding operator, $\RT$, by $\ell_q$ thresholding with $q=2/3$, $\LQuniv$, and by any shrinkage operator $\thrsig$ where $\sigma(1)=1/2$ and $\sigma$ satisfies the conditions of Lemma~\ref{lem:unifying_RC}. We now present our result for this setting:

\begin{theorem}\label{thm:RT_linearreg}
Suppose that $y = X\theta_0 + N(0,\sigma^2\mathbf{I}_n)$, where $\theta_0$ is $s_0$-sparse, and where $X\in\Xcal(\alpha,\beta,s)$, where $s = C\kappa s_0$ for some $C>2$. Suppose that $\thr$ is any $s$-sparse thresholding operator satisfying~\eqref{eqn:RC_rho_for_regr}.

Let $\widehat\theta_t$ be the estimate produced at step $t$ of the iterative thresholding algorithm~\eqref{eqn:RT_linearreg} initialized at some $s$-sparse $\widehat\theta_0\in\R^d$. Let $\tilde\theta_t \in\arg\min_{\theta\in\{\widehat\theta_1,\dots,\widehat\theta_t\}}\frac{1}{2n}\norm{y-X\theta}^2_2$, that is, $\tilde\theta_t$ is the best estimate seen before time $t$, relative to the loss function $\loss(\theta) = \frac{1}{2n}\norm{y-X\theta}^2_2$.

Then, for any $\delta>0$ and any $t\geq 1$,
\[\frac{1}{n}\norm{X(\tilde{\theta}_t-\theta_0)}^2_2 \leq \kappa\cdot \frac{28C \sigma^2 s_0\log(d) }{n} +\frac{12\sigma^2 \log(1/\delta)}{n} +  \left(\frac{1-1/\kappa}{1-2/C\kappa}\right)^t \cdot2\beta \norm{\widehat\theta_0 -\theta_0}^2_2,\]
with probability at least $1-\delta$.
\end{theorem}
\noindent  Since $t$ can be taken to be large (each iteration is very cheap), the dominant term is the first one, so we essentially have
\[\frac{1}{n}\norm{X(\tilde\theta_t - \theta_0)}^2_2 \lesssim \kappa\cdot \frac{\sigma^2 s_0\log(d) }{n}.\]
Comparing with the upper bound for iterative hard thresholding, we see that we now attains the ideal $\kappa$, rather than $\kappa^2$, scaling.

\paragraph{Comparison with Lasso}
The Lasso estimate of $\theta_0$, given by the convex optimization problem \[\widehat{\theta}\in\arg\min_{\theta\in\R^d}\left\{\frac{1}{2}\norm{y-X\theta}^2_2+\lambda\norm{\theta}_1\right\},\]
is proved in \citet{bickel2009simultaneous} to achieve a squared prediction error bounded as
\begin{equation}\label{eqn:lasso_upperbd}\frac{1}{n}\norm{X(\widehat{\theta}-\theta_0)}^2_2 \lesssim \kappa \cdot \frac{\sigma^2 s_0\log(d)}{n}\end{equation}
with a penalty parameter value $\lambda\sim \sigma\sqrt{\frac{\log(d)}{n}}$, under the assumption that $X\in \Xcal_{\textnormal{RE}}(\kappa,s_0)$. Compared with Lasso, due to Theorem~\ref{thm:RT_linearreg}, iterative thresholding algorithms with proper thresholding operators, for example the simple and efficient reciprocal thresholding, achieve the same squared prediction error bound. Moreover, both Lasso and iterative reciprocal thresholding method are guaranteed to give an estimator that is $\mathcal{O}(\kappa s_0)$ sparse (this sparsity level for Lasso is proved in \citet[Eqn. (7.9)]{bickel2009simultaneous}), and thus nearly match the computational lower bound with a gap in sparsity. An open question for future work is whether the larger sparsity level, i.e.~$\mathcal{O}(\kappa s_0)$ rather than $s_0$, is unavoidable to achieve the squared prediction error $\kappa \cdot \frac{\sigma^2 s_0\log(d)}{n}$, or whether there may be an $\mathcal{O}(s_0)$-sparse and computationally efficient estimator that achieves this bound.

\section{Discussion}

Relative concavity offers a framework for comparing theoretical properties of thresholding operators. Under this framework, we find a general class of optimal and near-optimal thresholding operators, among which is the new reciprocal thresholding operator, an alternative to hard and soft thresholding with tighter theoretical guarantees that is able to achieve better dependence on condition number for sparse and low-rank optimization problems. 

Nonetheless, many open questions remain for these problems. For example, our upper and lower bounds on $\lim_{t\rightarrow \infty}\loss(x_t)$ are proved relative to a broad class of functions satisfying (restricted) convexity and smoothness properties, with no underlying statistical model. In a statistical framework, we may be able to make additional assumptions, for instance, assuming that $\nabla\loss(y)$ is small at some highly sparse $y$ (e.g.~if $y$ is the true model parameter vector, while $\loss$ is the negative log-likelihood on the observed data)---is the relative concavity still necessary and sufficient for optimization guarantees, or would we observe different behavior of the various thresholding operators in this statistical setting? 

Relatedly, the relative concavity characterizes the restricted optimality guarantee of a thresholding operator on the worst-case objective function. 
In practice we may be more interested in the average-case convergence behavior of a thresholding operator, if the objective function $\loss$ arises from
some underlying statistical model or random process. Furthermore, how does the choice of the thresholding operator interact with modifications of the gradient descent algorithm, such as decreasing step size, choosing the step size via backtracking or another adaptive method, acceleration of the gradient descent step, replacing gradients with stochastic gradients, or using second-order information? We hope to address these directions in future work.

\subsection*{Acknowledgements}
R.F.B.~was partially supported by the National Science Foundation via grant DMS-1654076, and by an Alfred
P.~Sloan fellowship. The authors are grateful
to Chao Gao for helpful discussions and feedback on this work.

\bibliographystyle{plainnat}
\bibliography{iterativethresholding}

\appendix

\section{Proofs}

\subsection{Proofs of upper and lower bounds on convergence}

In this section, we prove our upper and lower bounds on convergence for the sparse setting, Theorems~\ref{thm:upperbd} and~\ref{thm:lowerbd}. The results for the matrix setting, Theorems~\ref{thm:matrix_upperbd} and~\ref{thm:matrix_lowerbd}, are proved identically, so we do not give those proofs here.

\begin{proof}[Proof of Theorem~\ref{thm:upperbd}]
Fix any $t\in\{1,\dots,T\}$. Since $x_{t-1}$ and $x_t$ are $s$-sparse by definition of the algorithm, and $\loss$ satisfies $(\alpha,s)$-RSC and $(\beta,s)$-RSM, we have
\begin{align*}
\loss(y)&\geq \loss(x_{t-1})+\langle \nabla \loss(x_{t-1}),y-x_{t-1}\rangle+\frac{\alpha}{2}\norm{y-x_{t-1}}^2_2,\\
\loss(x_t)&\leq \loss(x_{t-1})+\langle \nabla \loss(x_{t-1}),x_t-x_{t-1}\rangle+\frac{1}{2\eta_t}\norm{x_t-x_{t-1}}^2_2,
\end{align*}
where $\eta_t = \eta = \frac{1}{\beta}$ for the fixed step size algorithm~\eqref{eqn:iterative_thresh}, or $\eta_t$ is the adaptive step size defined in the algorithm~\eqref{eqn:iterative_thresh_eta_adaptive}---note that in this second case, since $\loss$ satisfies $(\beta,s)$-RSM, we see that $\eta_t \geq \frac{1}{\beta}$ since the step size is chosen by backtracking. Combining these two inequalities, we obtain
\begin{equation}\label{eqn:converge_step1}
\loss(x_t) - \loss(y)
 \leq \langle \nabla \loss(x_{t-1}), x_t - y\rangle +\frac{1}{2\eta_t}\norm{x_t-x_{t-1}}^2_2-  \frac{\alpha}{2}\norm{y-x_{t-1}}^2_2 .\end{equation}
We can also calculate
 \begin{align}
\notag&\frac{1}{2\eta_t}\norm{x_t - y}^2_2\\
\notag&=\frac{1}{2\eta_t}\norm{x_{t-1} - y}^2_2 -\frac{1}{2\eta_t}\norm{x_t - x_{t-1}}^2_2 + \frac{1}{\eta_t}\langle  x_{t-1} -x_t , y -x_t\rangle\\
\notag&=\frac{1}{2\eta_t}\norm{x_{t-1} - y}^2_2 -\frac{1}{2\eta_t}\norm{x_t - x_{t-1}}^2_2 + \frac{1}{\eta_t}\left\langle \left(x_{t-1} - \eta_t\nabla \loss(x_{t-1})\right) - x_t, y- x_t\right\rangle - \langle \nabla \loss(x_{t-1}),x_t-y\rangle \\
\label{eqn:converge_step2}&\leq\frac{1}{2\eta_t}\norm{x_{t-1} - y}^2_2 -\frac{1}{2\eta_t}\norm{x_t - x_{t-1}}^2_2 + \frac{1}{\eta_t} \cdot \RC_{s,\rho}(\thr)\cdot \norm{x_t - y}^2_2-\langle \nabla \loss(x_{t-1}),x_t-y\rangle,
\end{align}
where the last step applies the definition of restricted concavity, since $x_t = \thr\left(x_{t-1}-\eta_t\nabla \loss(x_{t-1})\right)$ by definition of the algorithm.

Combining steps~\eqref{eqn:converge_step1} and~\eqref{eqn:converge_step2}, then,
\[\loss(x_t) - \loss(y) \leq  \frac{1}{2\eta_t}\left[\left(1-\eta_t\alpha\right)\norm{x_{t-1} - y}^2_2 - \left(1 - 2\RC_{s,\rho}(\thr)\right)\norm{x_t - y}^2_2 \right].\]
Since $\eta_t\alpha\geq \frac{1}{\beta}\cdot \alpha = \frac{1}{\kappa}$, this implies
\[\loss(x_t) - \loss(y) \leq  \frac{1}{2\eta_t}\left[\left(1-\frac{1}{\kappa}\right)\norm{x_{t-1} - y}^2_2 - \left(1 - 2\RC_{s,\rho}(\thr)\right)\norm{x_t - y}^2_2 \right].\]
Taking a weighted sum over $t=1,\dots,T$, we obtain
\begin{align*}
&\sum_{t=1}^{T} 2\eta_t \left(\frac{1-1/\kappa}{1-2\RC_{s,\rho}(\thr)}\right)^{T-t} \cdot(\loss(x_t)-\loss(y))\\
&\leq \sum_{t=1}^{T} \left(\frac{1-1/\kappa}{1-2\RC_{s,\rho}(\thr)}\right)^{T-t} \cdot\left[\left(1-\frac{1}{\kappa}\right)\norm{x_{t-1} - y}^2_2 - \left(1 - 2\RC_{s,\rho}(\thr)\right)\norm{x_t - y}^2_2 \right]\\
&= \left(1 - 2\RC_{s,\rho}(\thr)\right) \cdot \sum_{t=1}^{T} \left[\left(\frac{1-1/\kappa}{1-2\RC_{s,\rho}(\thr)}\right)^{T-t+1} \norm{x_{t-1} - y}^2_2 - \left(\frac{1-1/\kappa}{1-2\RC_{s,\rho}(\thr)}\right)^{T-t} \norm{x_t - y}^2_2 \right]\\
&= \left(1 - 2\RC_{s,\rho}(\thr)\right) \cdot \left[\left(\frac{1-1/\kappa}{1-2\RC_{s,\rho}(\thr)}\right)^T \norm{x_0 - y}^2_2 -  \norm{x_T - y}^2_2 \right]\\
&\leq  \left(\frac{1-1/\kappa}{1-2\RC_{s,\rho}(\thr)}\right)^T \norm{x_0 - y}^2_2 ,
\end{align*}
where the next-to-last step simply cancels terms in the telescoping sum.
After rescaling, we have
\[\frac{\sum_{t=1}^{T} 2\eta_t\left(\frac{1-1/\kappa}{1-2\RC_{s,\rho}(\thr)}\right)^{T-t} \loss(x_t)}{\sum_{t=1}^{T} 2\eta_t\left(\frac{1-1/\kappa}{1-2\RC_{s,\rho}(\thr)}\right)^{T-t}} \leq \loss(y) + \frac{\left(\frac{1-1/\kappa}{1-2\RC_{s,\rho}(\thr)}\right)^T \norm{x_0 - y}^2_2 }{\sum_{t=1}^{T} 2\eta_t\left(\frac{1-1/\kappa}{1-2\RC_{s,\rho}(\thr)}\right)^{T-t}}.
\]
The left-hand side is a weighted average of $\loss(x_1),\loss(x_2),\dots,\loss(x_T)$, and is therefore lower-bounded by $\min_{t=1,\dots,T}\loss(x_t)$, while the denominator on the right-hand side is lower-bounded as
\[\sum_{t=1}^{T} 2\eta_t\left(\frac{1-1/\kappa}{1-2\RC_{s,\rho}(\thr)}\right)^{T-t}\geq 2\eta_T \geq \frac{2}{\beta}.\] After simplifying, we therefore have
\[\min_{t=1,\dots,T}\loss(x_t)\leq \loss(y) + \frac{\beta}{2}\cdot \left(\frac{1-1/\kappa}{1-2\RC_{s,\rho}(\thr)}\right)^T \cdot \norm{x_0 - y}^2_2 ,\]
as desired.
\end{proof}

\begin{proof}[Proof of Theorem~\ref{thm:lowerbd}]
By definition of $\RC_{s,\rho}(\thr)$, for any $\delta>0$, there exist some $s'$-sparse $y\in\R^d$ and some $z\in\R^d$ such that $x = \thr(z)\neq y$ and
\[\langle y - x, z-x\rangle \geq \RC_{s,\rho}(\thr)\cdot \norm{y-x}^2_2\cdot (1-\delta).\]
Let $U\in\R^{d\times d}$ be any orthogonal matrix with its first column equal to $\frac{y-x}{\norm{y-x}_2}$. We now define an objective function as
\[\loss(w) = -\beta \langle z-x, w-x\rangle + \frac{1}{2}(w-x)^\top UDU^\top (w-x)\text{ where }D=\left(\begin{array}{cccc}\alpha&0&\dots&0\\0&a_2&\dots&0\\\dots&\dots&\dots&\dots\\0&0&\dots&a_d\end{array}\right),\]
for some $a_2,\dots,a_d\in [\alpha,\beta]$.
Clearly, $\loss$ satisfies $(\alpha,s)$-RSC and $(\beta,s)$-RSM. Next, we can check that $\loss(x)=0$, while
\begin{align*}
\loss(y)& = -\beta \langle z-x, y-x\rangle + \frac{1}{2}(y-x)^\top UDU^\top (y-x)\\
&= -\beta \langle z-x, y-x\rangle +\frac{\alpha}{2}\norm{y-x}^2_2\\
&\leq -\beta \RC_{s,\rho}(\thr)\cdot \norm{y-x}^2_2\cdot (1-\delta) + \frac{\alpha}{2}\norm{y-x}^2_2\\
&= - \beta\norm{y-x}^2_2\cdot \left(\RC_{s,\rho}(\thr)\cdot (1-\delta) - \frac{1}{2\kappa}\right),
\end{align*}
where the first step uses the definition of $U$, while the inequality follows from the definition of $x,y,z$. Since $\gamma_{s,\rho}(\thr)>\frac{1}{2\kappa}$ by assumption, and $\delta$ can be chosen to be arbitrarily small, we therefore have $\loss(y)<0=\loss(x)$.

Finally, computing $\nabla\loss(w) = -\beta(z-w) + UDU^\top(w-x)$, suppose that we run the iterated thresholding algorithm~\eqref{eqn:iterative_thresh} with step size $\eta = \frac{1}{\beta}$, initialized at the point $x_0=x$. Since we have $\nabla\loss(x) = -\beta(z-x)$, the first update step is given by
\[x_1 =\thr\left(x - \frac{1}{\beta}\nabla \loss(x)\right) = \thr(z) = x.\]
This proves that $x$ is a stationary point of the algorithm---in other words, if the algorithm is initalized at $x_0=x$, then $x_t=x$ for all $t\geq 1$. Therefore, $\lim_{t\rightarrow\infty}\loss(x_t) = \loss(x) >\loss(y)$, as desired.
\end{proof}

\subsection{Proof for regularized minimization}
In this section we prove the result for the regularized rather than sparsity-constrained case given in Section~\ref{sec:softthresh}, i.e., using a proximal
map in place of a sparse thresholding operator.

\begin{proof}[Proof of Theorem~\ref{thm:prox_worstcase}]
Without loss of generality, take $\alpha=1$. 
Let $v,w$ be dense vectors satisfying $w\in\partial\reg(v)$, which are assumed to exist by the conditions of the theorem.
Define
\[\loss(z) = \frac{1}{2}\norm{z - (v + cw)}^2_2,\]
where $c$ is chosen to be large enough to satisfy
\[c^2\norm{w}^2_{\infty} > 2c \norm{w}_2\norm{v}_2 + \norm{v}^2_2.\]
This function is $\alpha$-strongly convex and $\beta$-smooth (since $\beta\geq \alpha\geq 1$), and therefore
satisfies $(\alpha,d)$-RSC and $(\beta,d)$-RSM.

Let $i\in\{1,\dots,d\}$ be the index of the largest-magnitude entry in $w$, and let $y = cw_i \cdot \mathbf{e}_i$, 
 where $\mathbf{e}_i$ is the vector with a $1$ in entry $i$ and zeros elsewhere. The condition on $c$ implies that
\[\loss(v) = c^2\norm{w}^2_2 = c^2 w_i^2 + c^2 \norm{w_{-i}}^2_2 \\
>c^2\norm{w_{-i}}^2_2 + 2c\norm{w_{-i}}_2\norm{v}_2 + \norm{v}^2_2 \geq\norm{cw_{-i} +v}^2_2 = \loss(y).\]

Next fix any $\lambda\geq0$. If $\norm{\widehat{x}_\lambda}_0=d$ then this proves our claim for this $\lambda$.
Otherwise, assume that $\norm{\widehat{x}_\lambda}_0<d$. By definition of $\widehat{x}_\lambda$, it must be the case
that 
\[\widehat{x}_\lambda \in \arg\min_{z\in\R^d}\left\{\frac{1}{2}\norm{z - (v+cw)}^2_2+ \lambda\reg(z)\right\} \quad \Rightarrow \quad \widehat{x}_{\lambda} = \textnormal{Prox}_{\lambda\reg}(v+cw).\] On the other hand, since $w\in\partial\reg(v)$,
this means that
\[v =\textnormal{Prox}_{c\reg}(v+cw) \quad \Rightarrow \quad v \in \arg\min_{z\in\R^d}\left\{\frac{1}{2}\norm{z - (v+cw)}^2_2+ c\reg(z)\right\}.\]
Thus we have
\[\loss(\widehat{x}_\lambda) + \lambda\reg(\widehat{x}_\lambda) \leq \loss(v) + \lambda\reg(v) \text{\quad and \quad}\loss(\widehat{x}_\lambda) + c\reg(\widehat{x}_\lambda) \geq \loss(v) + c\reg(v).\]
Rearranging terms, we obtain
\[\big(\loss(\widehat{x}_\lambda) - \loss(v)\big) \cdot (\lambda - c) \geq 0.\]
Furthermore, since  $v$ is dense but $\widehat{x}_{\lambda}$ is not, by our assumptions on the proximal map, this implies
that we must have $c<\lambda$, and therefore,
\[\loss(\widehat{x}_\lambda) \geq \loss(v) > \loss(y),\]
where the last step was proved previously. This completes the proof of the theorem.
\end{proof}

\subsection{Proofs for calculating relative concavity}
In this section we give the proofs for all lemmas from Sections~\ref{sec:RC} and~\ref{sec:matrix}, calculating upper and lower bounds on relative concavity in the vector and matrix setting.

\begin{proof}[Proof of Lemma~\ref{lem:RC_HT}]
Fix any $z\in\R^d$ and any $s'$-sparse $y\in\R^d$. Let $x=\HT(z)$. Let $S = \textnormal{Support}(x)$ and $S' = \textnormal{Support}(y)$.
We can write
\begin{equation}\label{eqn:HT_proof_start}
\frac{\inner{y-x}{z-x}}{\norm{y-x}^2_2}
=\frac{\inner{y_{S'\backslash S}}{z_{S'\backslash S}}}{\norm{y_{S'\backslash S}}^2_2 + \norm{(y-z)_S}^2_2},\end{equation}
since $x_S=z_S$ by definition of hard thresholding.
Next, let $\tau = \max_{i\not\in S} |z_i|$, i.e.~the $(s+1)$-st largest magnitude entry of $z$.
Then $|z_i|\geq \tau$ for all $i\in S$ by definition of the method, and so $|(y-z)_i|\geq \tau$ for all $i\in S\backslash S'$. Therefore,
\[\norm{(y-z)_S}^2_2\geq \tau^2 \cdot (s-\ell),\]
where $\ell = |S\cap S'|$. We also have
\[\inner{y_{S'\backslash S}}{z_{S'\backslash S}}\leq \norm{y_{S'\backslash S}}_2\cdot \tau\sqrt{s'-\ell},\]
since $\norm{z_{S'\backslash S}}_2\leq \tau \sqrt{|S'\backslash S|}\leq \tau\sqrt{s'-\ell}$. Combining everything and returning to~\eqref{eqn:HT_proof_start}, we have
\[\frac{\inner{y-x}{z-x}}{\norm{y-x}^2_2}
\leq \frac{\norm{y_{S'\backslash S}}_2\cdot \tau\sqrt{s'-\ell}}{\norm{y_{S'\backslash S}}^2_2 + \tau^2 \cdot (s-\ell)} \leq \max_{t\geq 0} \frac{t\sqrt{s'-\ell}}{t^2 + s-\ell},\]
where for the last step we consider $t = \frac{\norm{y_{S'\backslash S}}_2}{\tau}$. This quantity is maximized at $t = \sqrt{s-\ell}$, so we obtain
\[\frac{\inner{y-x}{z-x}}{\norm{y-x}^2_2}\leq \frac{\sqrt{s-\ell}\cdot\sqrt{s'-\ell}}{2(s-\ell)} = \frac{1}{2}\sqrt{\frac{s'-\ell}{s-\ell}}. \]
Finally, by definition, we must have $\ell\in\{0,1,\dots,s'\}$, so we obtain
\[\frac{\inner{y-x}{z-x}}{\norm{y-x}^2_2}\leq \max_{\ell\in\{0,1,\dots,s'\}} \frac{1}{2}\sqrt{\frac{s'-\ell}{s-\ell}} = \frac{1}{2} \sqrt{\rho},\]
where the maximum is obtained at $\ell=0$. This proves that $\RC_{s,\rho}(\HT)\leq \frac{\sqrt{\rho}}{2}$.

To prove a matching lower bound, consider $z = \mathbf{1}_d$. Then $x = \HT(z) = \mathbf{1}_S$, for some subset $S\subset\{1,\dots,d\}$ of cardinality $|S|=s$. Let $S'\subset\{1,\dots,d\}\backslash S$ be a disjoint set of cardinality $|S'|=s'$ (recall that we have assumed $s+s'\leq d$), and let $y = \frac{1}{\sqrt{\rho}} \cdot \mathbf{1}_{S'}$. Then
\[\frac{\inner{y-x}{z-x}}{\norm{y-x}^2_2} = \frac{\frac{1}{\sqrt{\rho}} \cdot s'}{\frac{1}{\rho}\cdot s' + s} = \frac{\sqrt{\rho}}{2},\]
thus proving that $\RC_{s,\rho}(\HT)\geq \frac{\sqrt{\rho}}{2}$.
\end{proof}

\begin{proof}[Proof of Lemma~\ref{lem:RC_contin}]
We consider two cases. If there exists $z\neq 0$ such that $\thr(z)=0$, then fix any such $z$ and fix an index $i$ such that $z_i \neq 0$. Let $y = \epsilon\cdot \textnormal{sign}(z_i)\cdot \mathbf{e}_i$, where $\mathbf{e}_i$ is the vector with a $1$ in entry $i$ and zeros elsewhere. $y$ is $s'$-sparse since $s'\geq 1$. Then
\[\frac{\langle y-\thr(z),z-\thr(z)}{\norm{y-\thr(z)}^2_2}
=\frac{\langle y,z\rangle}{\norm{y}^2_2}
= \frac{\epsilon |z_i|}{\epsilon^2} = \frac{|z_i|}{\epsilon}.
\]
Since $|z_i|>0$ and $\epsilon>0$ can be taken to be arbitrarily small, this shows that $\RC_{s,\rho}(\thr) = \infty\geq 1$.

On the other hand, if $\thr(z)\neq 0$ for any $z\neq 0$, then define $g: \mathbb{S}^{d-1}\rightarrow \mathbb{S}^{d-1}$ by $g(x)=\frac{\thr(x)}{\norm{\thr(x)}_2}$, where $\mathbb{S}^{d-1}$ is the unit sphere in $\R^d$. Since $\norm{\thr(x)}_2$ is a continuous function on a compact space and takes only positive values, $\norm{\thr(x)}_2$ is lower-bounded by a positive value, which then implies $g$ is continuous. Since $g(x)$ inherits the sparsity of $\thr(x)$ for all $x$, we see that $g(\mathbb{S}^{d-1})\subset \mathbb{S}^{d-1}\backslash\{x_0\}$, where $x_0 =\mathbf{1}_d/\sqrt{d}$ is a dense point on the sphere. Now let $h$ be a homeomorphism from $\mathbb{S}^{d-1}\setminus\{x_0\}$ to $\mathbb{R}^{d-1}$ (for example, take $h$ to be the stereographic projection from the point $x_0$). Then $h\circ g: \mathbb{S}^{d-1}\rightarrow \mathbb{R}^{d-1}$ is continuous. By the Borsuk-Ulam theorem, there exist two antipodal point being mapped to the same point, i.e.~there exists $z\in\mathbb{S}^{d-1}$ such that $h\circ g(z)=h\circ g(-z)$, and thus $g(z)=g(-z)$ since $h$ is bijective. Now, there are two possibities---either $\inner{z}{g(z)}\leq 0$, or alternately $\inner{z}{g(z)}>0$ in which case $\inner{-z}{g(z)}=\inner{-z}{g(-z)}<0$. Replacing $z$ with $-z$ if needed, then, we have some $z\in\mathbb{S}^{d-1}$ such that $\inner{z}{g(z)}\leq 0$. Then by definition of $g$, we have $\inner{z}{\thr(z)}\leq 0$. Setting $y=\mathbf{0}_d$, we then calculate
\[\frac{\inner{y-\thr(z)}{z-\thr(z)}}{\norm{y-\thr(z)}^2_2} =\frac{\norm{\thr(z)}^2_2 - \inner{z}{\thr(z)}}{\norm{\thr(z)}^2_2} \geq\frac{\norm{\thr(z)}^2_2 - 0}{\norm{\thr(z)}^2_2} = 1,\]
proving that $\RC_{s,\rho}(\thr)\geq 1$, as desired.
\end{proof}

\begin{proof}[Proof of Lemmas~\ref{lem:RC_RT} and~\ref{lem:RC_LQ}]
These lemmas are special cases of the general result, Lemma~\ref{lem:unifying_RC}, proved below.
\end{proof}

\begin{proof}[Proof of Lemma~\ref{lem:RC_lowerbd}]
Let $z = \mathbf{1}_d$ and let $x=\thr(z)$. Let $S=\textnormal{Support}(x)$, with $|S|\leq s$, and let $S'\subset \{1,\dots,d\}\backslash S$ be any set disjoint from $S$, with cardinality $|S'| = s'$ (recall that we have assumed $s+s'\leq d$). Let $y = t \cdot \mathbf{1}_{S'}$, where
\[t = \frac{r}{\rho}\left(1 - r + \sqrt{r^2-2r+ 1+\rho}\right), \text{ for }r = \frac{\norm{x}_2}{\sqrt{s}}.\]
Then $\norm{y}_0=s'$, and we can calculate
\begin{align}
\frac{\inner{y-x}{z-x}}{\norm{y-x}^2_2}
\notag&=\frac{\inner{y}{z} - \inner{x}{z}+\norm{x}^2_2}{\norm{y}^2_2 + \norm{x}^2_2}\text{ since $x,y$ have disjoint supports $S$ and $S'$}\\
\notag&=\frac{t\cdot s' - \inner{x}{\mathbf{1}_d}+\norm{x}^2_2}{t^2\cdot s' + \norm{x}^2_2}\text{ by definition of $y$ and $z$}\\
\label{eqn:RC_lowerbd_step}&\geq\frac{t\cdot s' - \sqrt{s}\cdot \norm{x}_2+\norm{x}^2_2}{t^2\cdot s' + \norm{x}^2_2}\text{ since $x$ is $s$-sparse}\\
\notag&=\frac{t\cdot s' - s \cdot r + s\cdot r^2}{t^2\cdot s' + s\cdot r^2}\\
\notag&=\frac{t\cdot \rho - r + r^2}{t^2\cdot \rho + r^2}.
\end{align}
Plugging in the value of $t$ that we chose above, we continue:
\begin{align*}
\frac{\inner{y-x}{z-x}}{\norm{y-x}^2_2}
&\geq\frac{\frac{r}{\rho}\left(1 - r + \sqrt{r^2-2r+ 1+\rho}\right)\cdot \rho - r + r^2}{\frac{r^2}{\rho^2}\left(1 - r + \sqrt{r^2-2r+ 1+\rho}\right)^2\cdot \rho + r^2}\\
&=\frac{r \sqrt{r^2-2r+1+\rho}}{\frac{r^2}{\rho}\left(1 - r + \sqrt{r^2-2r+ 1+\rho}\right)^2 + r^2 }\\
&=\frac{\rho }{r\left(2\sqrt{r^2-2r+1+\rho} + 2(1-r) \right)},
\end{align*}
where the last few steps are just simplifying the expression.
Next, we consider the denominator. It can easily be verified that
\[r\left(2\sqrt{r^2-2r+1+\rho} + 2(1-r) \right)\leq 1+\rho\]
for all $r\geq 0$, which we check by verifying that the left-hand side is maximized when $r=\frac{1+\rho}{2}$. Therefore,
\[\frac{\inner{y-x}{z-x}}{\norm{y-x}^2_2} \geq \frac{\rho}{1+\rho},\]
which proves that
\[\RC_{s,\rho}(\thr)\geq \frac{\rho}{1+\rho},\]
as desired.
\end{proof}

\begin{proof}[Proof of Lemma~\ref{lem:unifying_RC}]
We first show the upper bound. Fix any $z\in\R^d$ and any $s'$-sparse $y\in\R^d$. Let $x=\thrsig(z)$. Let $S = \textnormal{Support}(x)$ and $S' = \textnormal{Support}(y)$, and let $\ell = |S\cap S'|$.  Then we have
\begin{align*}
\frac{\inner{y-x}{z-x}}{\norm{y-x}^2_2}
&=\frac{\inner{(y-x)_{S'}}{(z-x)_{S'}} - \inner{x_{S\backslash S'}}{(z-x)_{S\backslash S'}}}{\norm{(y-x)_{S'}}^2_2  + \norm{x_{S\backslash S'}}^2_2}\\
&\leq\frac{\norm{(y-x)_{S'}}_2\norm{(z-x)_{S'}}_2 - \inner{x_{S\backslash S'}}{(z-x)_{S\backslash S'}}}{\norm{(y-x)_{S'}}^2_2  + \norm{x_{S\backslash S'}}^2_2}.
\end{align*}
Let $\tau = \max_{i\not\in S} |z_i|$, i.e. the thresholding level. Due to the definition of $\thrsig$ and the assumptions on $\sigma$, it is direct to verify the following bounds: if $i\not\in S$, then $|(z-x)_i|=|z_i|\leq \tau$; if $i\in S$, then $|(z-x)_i| \leq \tau\sigma(1)$, $|x_i|\geq \tau(1-\sigma(1))$, and $x_i(z-x)_i\geq \tau^2 \sigma(1)\big(1-\sigma(1)\big)$. Plugging these bounds back in our calculation above, we get:
\begin{multline*}
\frac{\inner{y-x}{z-x}}{\norm{y-x}^2_2}
\leq\frac{\norm{(y-x)_{S'}}_2\sqrt{(s'-\ell)\cdot \tau^2+\ell\cdot \tau^2\sigma(1)^2}-(s-\ell)\cdot \tau^2\sigma(1)\big(1-\sigma(1)\big)}{\norm{(y-x)_{S'}}_2^2+(s-\ell)\cdot \tau^2\big(1-\sigma(1)\big)^2}\\
\leq\max_{t\geq 0}\frac{t\sqrt{\frac{s'-\ell}{s-\ell}+\frac{\ell}{s-\ell}\cdot \sigma(1)^2}-\sigma(1)\big(1-\sigma(1)\big)}{t^2+\big(1-\sigma(1)\big)^2},
\end{multline*}
where the last step holds by considering $t=\frac{\norm{(y-x)_{S'}}_2}{\tau\sqrt{s-\ell}}$. Next, we can calculate
\[
\max_{\ell \in\{0,\dots,s'\}}\sqrt{\frac{s'-\ell}{s-\ell}+\frac{\ell}{s-\ell}\cdot\sigma(1)^2}
=\max_{\ell \in\{0,\dots,s'\}}\sqrt{\frac{s' - \big(1-\sigma(1)^2\big) \ell}{s-\ell}} = \sqrt{\frac{\rho}{\min\big\{1,(1-\rho)/\sigma(1)^2\big\}}},\]
where the maximum is attained at $\ell=0$ if $\rho\leq 1-\sigma(1)^2$, and at $\ell=s'$ otherwise.
It therefore follows that
\begin{multline}\label{eqn:max_RC_thrsig}
\frac{\inner{y-x}{z-x}}{\norm{y-x}^2_2}\leq \max_{t\geq 0}\frac{t\sqrt{\frac{\rho}{\min\{1,(1-\rho)/\sigma(1)^2\}}}-\sigma(1)\big(1-\sigma(1)\big)}{t^2+\big(1-\sigma(1)\big)^2}\\
=\frac{\frac{\rho}{\min\{1,(1-\rho)/\sigma(1)^2\}}}{2\sigma(1)\big(1-\sigma(1)\big)\left(1 + \sqrt{1 + \frac{\rho/\sigma(1)^2}{\min\{1,(1-\rho)/\sigma(1)^2\}}}\right)},
\end{multline}
where to compute the last step we can check that the maximum is achieved at
\[t = \frac{\sigma(1)\big(1-\sigma(1)\big) + \sqrt{\sigma(1)^2\big(1-\sigma(1)\big)^2 + \frac{\rho\big(1-\sigma(1)\big)^2}{\min\{1,(1-\rho)/\sigma(1)^2\}}}}{\sqrt{\frac{\rho}{\min\{1,(1-\rho)/\sigma(1)^2\}}}}.\] This proves the upper bound. To prove the lower bound, we simply choose $y$ and $z$ so that the inequalities above become equalities. Set $z = \mathbf{1}_d$ and $x=\RT(z)$, and let $S = \textnormal{Support}(x)$. Due to the definition of $\thrsig$, we see that $x = (1-\sigma(1))\cdot\mathbf{1}_S$. To construct $y$, we consider two cases. If $\rho\leq 1-\sigma(1)^2$, we let $S'\subset\{1,\dots,d\}\backslash S$ be any set disjoint from $S$ with cardinality $|S'|=s'$ (recall that $s+s'\leq d$ by assumption). Then let $y=\frac{t}{\sqrt{\rho}}\cdot\mathbf{1}_{S'}$, where $t\geq 0$ is arbitrary, so that we have
\[\frac{\inner{y-x}{z-x}}{\norm{y-x}^2_2} = \frac{\frac{t}{\sqrt{\rho}}\cdot s' - \sigma(1)\big(1-\sigma(1)\big)\cdot s}{\frac{t^2}{\rho}\cdot s' + \big(1-\sigma(1)\big)^2\cdot s} = \frac{t\sqrt{\rho} - \sigma(1)\big(1-\sigma(1)\big)}{t^2 + \big(1-\sigma(1)\big)^2}.\]
Alternately, if $\rho>1-\sigma(1)^2$, let $S'\subset S$ be any set of cardinality $|S'|=s'$, and set $y =\left(1-\sigma(1)+t\sqrt{\frac{1-\rho}{\rho}}\right)\cdot\mathbf{1}_{S'}$, where again $t>0$ is arbitrary. For this second case, we calculate
\[\frac{\inner{y-x}{z-x}}{\norm{y-x}^2_2} = \frac{t\sigma(1)\sqrt{\frac{1-\rho}{\rho}}\cdot s' - \sigma(1)\big(1-\sigma(1)\big)\cdot (s-s')}{t^2\cdot\frac{1-\rho}{\rho}\cdot s' + \big(1-\sigma(1)\big)^2\cdot (s-s')} = \frac{t\sqrt{\frac{\rho}{(1-\rho)/\sigma(1)^2}} - \sigma(1)\big(1-\sigma(1)\big)}{t^2 + \big(1-\sigma(1)\big)^2}.\]
Combining the two cases, and recalling that $t\geq 0$ is arbitrary, we see that
\[\RC_{s,\rho}(\thrsig)\geq \max_{t\geq 0}\frac{t\sqrt{\frac{\rho}{\min\{1,(1-\rho)/\sigma(1)^2\}}}-\sigma(1)\big(1-\sigma(1)\big)}{t^2+\big(1-\sigma(1)\big)^2} ,\]
which matches the upper bound calculated in~\eqref{eqn:max_RC_thrsig} above.
\end{proof}

\begin{proof}[Proof of Lemma~\ref{lem:RC_lowerbd_matrix}]
Without loss of generality, let $n\geq m$. Let $Z=\left(\begin{array}{c}\mathbf{I}_m \\\mathbf{0}_{(n-m)\times m}\end{array}\right)$, and let $X=\mthr(Z)$. Let $X=UDV^\top$ be a singular value decomposition of $X$, with $U\in\R^{n\times s},V\in\R^{m\times s}$. Let $V_{\perp}\in\R^{m\times s'}$ be an orthonormal matrix that is orthogonal to $V$ (recall that $s+s'\leq m$ by assumption), and let $Y=t\cdot \left(\begin{array}{c}V_{\perp}V_{\perp}^\top \\ \mathbf{0}_{(n-m)\times m}\end{array}\right)$, for some $t\geq 0$.
Then $\rank(Y) = s'$, and we can calculate
\begin{align*}
\frac{\inner{Y-X}{Z-X}}{\fronorm{Y-X}^2}
&=\frac{\inner{Y}{Z} - \inner{X}{Z}+\fronorm{X}^2}{\fronorm{Y}^2 + \fronorm{X}^2}\text{ since $X$ and $Y$ have orthogonal row spaces by def.~of $V_{\perp}$}\\
&=\frac{t\cdot s' - \norm{X}_*\norm{Z}+\fronorm{X}^2}{t^2\cdot s' + \fronorm{X}^2}\text{ by def.~of $Y$ and $Z$ (here $\norm{\cdot}_*$ is the nuclear norm)}\\
&\geq\frac{t\cdot s' - \sqrt{s}\cdot \fronorm{X}+\fronorm{X}^2}{t^2\cdot s' + \fronorm{X}^2}\text{ since $\rank(X)\leq s$}.
\end{align*}
Comparing to~\eqref{eqn:RC_lowerbd_step}, we see that the remainder of the argument is identical to the proof of Lemma~\ref{lem:RC_lowerbd}.
\end{proof}

\begin{proof}[Proof of Lemma~\ref{lem:matrix_RC}]
First, fix any $y,z\in\R^d$. Let $Y=\textnormal{diag}(y)$ and  $Z=\textnormal{diag}(z)$, so that $\mthr(Z) = \textnormal{diag}(\thr(z))$ and $\rank(Y)=\norm{y}_0$. Then we trivially have $\frac{\inner{Y-\mthr(Z)}{Z-\mthr(Z)}}{\fronorm{Y-\mthr{Z}}^2} = \frac{\inner{y-\thr(z)}{z-\thr(z)}}{\norm{y-\thr{z}}^2_2}$, and maximizing over all $y,z$ yields the restricted concavity, $\RC_{s,\rho}(\thr)$. This proves that $\mRC_{s,\rho}(\mthr)\geq \RC_{s,\rho}(\thr)$.

 Next we show the reverse inequality. Consider any $Y,Z\in\R^{n\times m}$ with $\rank(Y)\leq s'$, and let $X=\mthr(Z) =U\cdot\textnormal{diag}(\thr(d))\cdot V^\top$, where $Z=U\cdot\textnormal{diag}(d)\cdot V^\top$ is the singular value decomposition. We want to prove the claim that
\[\inner{Y-X}{Z-X}\leq \RC_{s,\rho}(\thr)\fronorm{Y-X}^2.\]
In other words, defining
\[\mathsf{h}(Y) =  \RC_{s,\rho}(\thr)\fronorm{Y-X}^2 - \inner{Y-X}{Z-X} = \RC_{s,\rho}(\thr)\Fronorm{Y - \left(X + \frac{Z-X}{2\RC_{s,\rho}(\thr)}\right)}^2 -\frac{\fronorm{Z-X}^2}{4\RC_{s,\rho}(\thr)},\]
we'd like to show that $\mathsf{h}(Y)\geq 0$ for all rank-$s'$ matrices $Y$. Now, by definition of $X$, we can see that $U$ and $V$ are the left and right singular vector matrices for $X + \frac{Z-X}{2\RC_{s,\rho}(\thr)}$, and therefore $\mathsf{h}(Y)$ is minimized by some matrix $Y$ of the form $Y=U\cdot \textnormal{diag}(y)\cdot V^\top$, for some $s'$-sparse vector $y$. Now, for any matrix of this form, we have
\begin{multline*}\inner{Y-X}{Z-X}=\inner{U\cdot \textnormal{diag}(y)\cdot V^\top - U\cdot \textnormal{diag}(\thr(d))\cdot V^\top}{U\cdot \textnormal{diag}(d)\cdot V^\top - U\cdot \textnormal{diag}(\thr(d))\cdot V^\top}\\
=\inner{y-\thr(d)}{d-\thr(d)}\leq \RC_{s,\rho}(\thr)\norm{y-\thr(d)}^2_2 = \RC_{s,\rho}(\thr)\fronorm{Y-X}^2,\end{multline*}
by using the definition of relative concavity for sparse vectors. This proves that
\[\min_{\rank(Y)\leq s'}\mathsf{h}(Y) = \min_{Y=U\cdot\textnormal{diag}(y)\cdot V^\top,\norm{y}_0\leq s}\mathsf{h}(Y) \geq 0,\]
thus proving that $\mRC_{s,\rho}(\mthr)\leq \RC_{s,\rho}(\thr)$, as desired.
\end{proof}

\subsection{Proofs for prediction error in linear regression}
In this section we prove our prediction error bounds for the linear regression setting.

\begin{proof}[Proof of Theorem~\ref{thm:RT_linearreg}]
Since $s=C\kappa s_0$ and so our sparsity ratio is $\rho = \frac{1}{C\kappa}\leq \frac{1}{2}$, Lemma~\ref{lem:unifying_RC} with the conditions on $\sigma$ proves that $\RC_{s,\rho}(\thrsig)\leq \rho = \frac{1}{C\kappa}$. Since this is strictly smaller than $\frac{1}{2\kappa}$, Theorem~\ref{thm:upperbd} proves that
\[\loss(\tilde\theta_t) \leq \loss(\theta_0) + \left(\frac{1-1/\kappa}{1-2/C\kappa}\right)^t \cdot\frac{\beta}{2}\norm{\widehat\theta_0 -\theta_0}^2_2.\]
Next, recalling the definition of $\loss(\theta)$, this is equivalent to
\[\frac{1}{2n}\norm{\sigma z - X(\tilde\theta_t - \theta_0)}^2_2 \leq \frac{1}{2n}\norm{\sigma  z}^2_2 + \left(\frac{1-1/\kappa}{1-2/C\kappa}\right)^t \cdot\frac{\beta}{2}\norm{\widehat\theta_0 -\theta_0}^2_2,\]
where $z\sim N(0,\mathbf{I}_n)$ and $y=X\theta_0 + \sigma z$. Rearranging terms,
\[\frac{1}{2n}\norm{X(\tilde\theta_t - \theta_0)}^2_2 \leq \frac{\sigma }{n}\inner{z}{X(\widehat\theta_t - \theta_0)} + \left(\frac{1-1/\kappa}{1-2/C\kappa}\right)^t \cdot\frac{\beta}{2}\norm{\widehat\theta_0 -\theta_0}^2_2.\]
Now, by Lemma~\ref{lem:sparse_eig} below, with probability at least $1-\delta$, we have
\begin{multline*}\inner{z}{X(\tilde\theta_t - \theta_0)}\leq \norm{X(\tilde\theta_t - \theta_0)}_2 \cdot \sqrt{7s\log(d) + 3\log(1/\delta)}\\
\leq \frac{1}{4\sigma}\norm{X(\tilde\theta_t - \theta_0)}^2_2 + \sigma (7s\log(d) + 3\log(1/\delta)),\end{multline*}
and so combining everything,
\[\frac{1}{2n}\norm{X(\tilde\theta_t - \theta_0)}^2_2 \leq \frac{1}{4n}\norm{X(\tilde\theta_t - \theta_0)}^2_2 + \sigma^2\cdot \frac{7s\log(d) + 3\log(1/\delta)}{n} + \left(\frac{1-1/\kappa}{1-2/C\kappa}\right)^t \cdot\frac{\beta}{2}\norm{\widehat\theta_0 -\theta_0}^2_2.\]
Rearranging terms, then,
\[\frac{1}{n}\norm{X(\tilde\theta_t - \theta_0)}^2_2 \leq \sigma^2 \cdot \frac{28s\log(d) + 12\log(1/\delta)}{n} + \left(\frac{1-1/\kappa}{1-2/C\kappa}\right)^t \cdot2\beta \norm{\widehat\theta_0 -\theta_0}^2_2.\]
Plugging in $s = C\kappa s_0$, this proves the theorem.
\end{proof}

\begin{lemma}\label{lem:sparse_eig}
Fix any sparsity level $s$, dimension $d\geq 3$, and sample size $n$. Fix any $s$-sparse $\theta_0\in\R^d$, and any matrix $X\in\R^{n\times d}$ such that $X\in\Xcal(\alpha,\beta,s)$ for some parameters $1\leq\alpha\leq \beta$. Let $z\sim N(0,\mathbf{I}_n)$. Then for any $\delta>0$,
\[\mathbb{P}\left\{\inner{z}{X(\theta - \theta_0)}\leq \norm{X(\theta-\theta_0)}_2\cdot \sqrt{7s\log(d) + 3\log(1/\delta)}\text{ for all $s$-sparse $\theta\in\R^d$}\right\} \geq 1-\delta.\]
\end{lemma}
\begin{proof}[Proof of Lemma~\ref{lem:sparse_eig}]
Let $A_0\subset\{1,\dots,d\}$ be the support of $\theta_0$. We take a union bound over all sets $A\subset\{1,\dots,d\}$ of size $|A|=s$. First, for any fixed $A$, let $U^A\in\R^{n\times |A\cup A_0|}$ be an orthogonal basis for the column space of $X_{A\cup A_0} = (X_{ij}){j\in A\cup A_0}\in\R^{n\times |A\cup A_0|}$. Then
\begin{multline*}
\inner{z}{X(\theta - \theta_0)} = \inner{z}{X_{A\cup A_0}(\theta - \theta_0)_{A\cup A_0}} = \inner{U^AU^A{}^\top z}{X_{A\cup A_0}(\theta - \theta_0)_{A\cup A_0}}\\
\leq \norm{U^AU^A{}^\top z}_2\norm{X_{A\cup A_0}(\theta - \theta_0)_{A\cup A_0}}_2 =  \norm{U^A{}^\top z}_2\norm{X_{A\cup A_0}(\theta - \theta_0)_{A\cup A_0}}_2.
\end{multline*}
Next, $\norm{U^A{}^\top z}^2_2\sim \chi^2_{|A\cup A_0|}\leq \chi^2_{2s}$. By \citet[Lemma 1]{laurent2000adaptive}, then,
\[\mathbb{P}\left\{\norm{U^A{}^\top z}^2_2 \geq 2s + 2\sqrt{2s t} + 2t\right\} \leq e^{-t}. \]
Taking $t = \log(d^s / \delta)$, we see that
\[\max_{|A|=s}\norm{U^A{}^\top z}^2_2 \leq 2s + 2\sqrt{2s\log(d^s / \delta)} + 2\log(d^s/\delta)\leq 7s \log(d) + 3\log(1/\delta)\]
with probability at least $1-\delta$, proving the lemma.
\end{proof}

\end{document}